\documentclass[12pt,draftcls,journal,onecolumn]{IEEEtran}

\usepackage{amssymb,amsthm, amsmath,latexsym}
\usepackage{graphicx}
\usepackage{mathrsfs}
\usepackage{amsfonts}
\usepackage{amssymb}
\usepackage{longtable}
\usepackage{amsmath}
\usepackage{setspace}
\usepackage{caption}
\usepackage[figuresright]{rotating}
\usepackage[misc]{ifsym}
\usepackage{bbm}
\usepackage{makecell}
\usepackage{arydshln}
\usepackage{supertabular}
\usepackage{booktabs}
\usepackage{color}

\newtheorem{theorem}{Theorem}
\newtheorem{lemma}[theorem]{Lemma}
\newtheorem{remark}[theorem]{Remark}

\newtheorem{corollary}[theorem]{Corollary}

\newtheorem{example}[theorem]{Example}

\newcommand{\rank}{{\mathrm{rank}}}

\newcommand{\C}{{\mathcal{C}}}

\newcommand{\bF}{ {\mathbb F}}

\newcommand{\tabincell}[2]{\begin{tabular}{@{}#1@{}}#2\end{tabular}}

\usepackage{blindtext}

\ifCLASSINFOpdf

\else

\fi

\begin{document}
\title{Several new infinite families of NMDS codes with arbitrary dimensions supporting $t$-designs
\thanks{
This work was supported in part by the National Natural Science Foundation of China under Grant Numbers 62272148 and in part by the Natural Science Foundation of Hubei Province of China under Grant Number 2023AFB847 and the Sunrise Program of Wuhan under the Grant Number 2023010201020419. The corresponding author is Dabin Zheng.
}
}

\author{Yaozong Zhang,\,\,\, Dabin Zheng,\,\,\, Xiaoqiang Wang,\,\,\, Wei Lu {\thanks{Y. Zhang, D. Zheng, X. Wang and W. Lu are with the Hubei Key Laboratory of Applied Mathematics, Faculty of Mathematics and Statistics, Hubei University, Wuhan 430062, China
(e-mail:zyzsdutms@163.com; dzheng@hubu.edu.cn; waxiqq@163.com; weilu23@hubu.edu.cn), and D. Zheng is also with the Key Laboratory of Intelligent Sensing System and Security, Ministry of Education, Hubei University, Wuhan 430062, China.}}
 }

\maketitle
\begin{abstract}
Near maximum distance separable (NMDS) codes, where both the code and its dual are almost maximum distance separable, play pivotal roles in combinatorial design theory and cryptographic applications. Despite progress in fixed dimensions (e.g., dimension 4 codes by Ding and Tang \cite{Ding2020}), constructing NMDS codes with arbitrary dimensions supporting $t$-designs ($t\geq 2$) has remained open. In this paper, we construct two infinite families of NMDS codes over $\mathbb{F}_q$ for any prime power $q$ with flexible dimensions and determine their weight distributions. Further, two additional families with arbitrary dimensions over $\mathbb{F}_{2^m}$ supporting $2$-designs and $3$-designs, and their weight distributions are obtained. Our results fully generalize prior fixed-dimension works~\cite{DingY2024,Heng2023,Heng20231,Xu2022}, and affirmatively settle the Heng-Wang conjecture \cite{Heng2023} on the existence of NMDS codes with flexible parameters supporting $2$-designs.
\end{abstract}

\section{Introduction}\label{sec-intro}

Throughout this paper, let $\mathbb{F}_q$ denote the finite field with $q$ elements, where $q=p^m$ is a prime power, $m \geq 3$ if $p=2$ and $m \geq 2$ otherwise.
Let $\mathbb{F}_q^*$ represent the multiplicative group of $\mathbb{F}_q$, and $\mathbb{F}_q^n$ the $n$-dimensional vector space over $\mathbb{F}_q$. A \emph{linear code} $\mathcal{C}$ with parameters $[n,k]_q$ is defined as a $k$-dimensional linear subspace of $\mathbb{F}_q^n$. A  \emph{generator matrix} for an $[n, k]_q$ code $\mathcal{C}$ is any $k\times n$ matrix $G$ whose rows form a basis for $\mathcal{C}$. In general, there are many generator matrices
for a code. The dual code of $\mathcal{C}$ is given by
$
\mathcal{C}^\perp = \left\{\mathbf{v} \in \mathbb{F}_q^n : \langle \mathbf{c}, \mathbf{v} \rangle = 0 \text{ for all } \mathbf{c} \in \mathcal{C}\right\},
$
where $\langle \mathbf{c}, \mathbf{v} \rangle$ denotes the Euclidean inner product of $\mathbf{c}$ and $\mathbf{v}$. The generator matrix of a linear code is the \emph{parity-check matrix} of its dual code. For a codeword $\mathbf{c}=(c_1, c_2, \cdots, c_n) \in \mathcal{C}$, its support is defined as $\operatorname{supp}(\mathbf{c}) = \{1 \leq i \leq n : c_i \neq 0\}$, and  $\operatorname{wt}(\mathbf{c}) = \#\operatorname{supp}(\mathbf{c})$ is the Hamming weight of $\mathbf{c}$, and the \emph{minimum Hamming distance} $d(\mathcal{C})$ is the minimal Hamming weight among all non-zero codewords. The \emph{weight distribution} of $\mathcal{C}$ is characterized by the sequence $(A_0, A_1, \ldots, A_n)$, and the polynomial $ A(x) = 1 + A_1x + A_2x^2 + \cdots + A_nx^n$ is called the \emph{weight enumerator} of $\mathcal{C}$, where $A_i$ denotes the number of codewords of weight $i$ in $\mathcal{C}$.

\subsection{NMDS codes supporting $t$-designs}
For an $[n,k,d]_q$ linear code, the classical Singleton bound states that $d \leq n - k + 1$ \cite{Huffman2003}. A code achieving equality is termed \emph{maximum distance separable (MDS)}, whose dual code is also MDS. When $d = n - k$, the code is called \emph{almost maximum distance separable (AMDS)}, and if its dual remains AMDS, it qualifies as \emph{near maximum distance separable (NMDS)}. First introduced by Dodunekov and Landjev \cite{Dodunekov1995}, NMDS codes find significant applications in finite geometry \cite{Dodunekov1995}, combinatorial designs \cite{Ding2020,Heng2023,Heng20231,Xu2022}, and cryptography \cite{Simos2012,Zhou2009}.
Recent years have witnessed significant advances in NMDS code constructions. Li and Heng \cite{Li2023} proposed infinite families of NMDS codes with lengths $q+1 \leq n \leq q+4$ and dimension 3 over ${\mathbb{F}_{2^m}}$ ($m\geq3$), employing specialized generator matrices while deriving their weight enumerators. Fan's work \cite{Fan2024} achieved variable-length NMDS codes ($n=q+m+2$, dimension 3) over ${\mathbb{F}_{2^m}}$ ($m$ odd) through strategic additions of $m$ projective points to maximum arcs in $PG(2,q)$. For parameters $1\leq N<\frac{q+1}{2}$ (even) and odd $q$, Heng et al. \cite{Heng20221} constructed NMDS codes with length $\frac{q+1}{N}+1$ and dimension 3 using cyclic subgroups of $\mathbb{F}_{q^2}^*$. Extended results in \cite{DingY2024,Heng2022} established NMDS families with length $n\geq q+1$ and dimension 4. More constructions for NMDS codes with fixed dimensions are documented in \cite{Ding2020,Heng2023,Heng20231,Wang2021,Xu2024,Xu2023,Xu20241} and the references therein.

Let $t,k,n$ be positive integers with $1 \leq t \leq k \leq n$. A \emph{$t$-$(n,k,\lambda)$ design} is a pair $(\mathcal{P}, \mathcal{B})$, where $\mathcal{P}$ is an $n$-set of points and $\mathcal{B}$ a collection of distinct $k$-subsets (blocks) satisfying that every $t$-subset lies in exactly $\lambda$ blocks. The design is \emph{simple} if all blocks are distinct, and \emph{trivial} if $k=t~{\rm or}~n$.
For any $t$-$(n,k,\lambda)$ design, the parameter relation
$
  \binom{n}{t}\lambda = \binom{k}{t}b
$
holds with $b = \#\mathcal{B}$ \cite{Huffman2003}. Moreover, its \emph{complementary design} $(\mathcal{P}, \mathcal{B}^c)$ forms a $t$-$(n,n-k,\lambda^c)$ design, where
\begin{equation}\label{eq(1.1)}
  \lambda^c = \lambda \cdot \binom{n-t}{k}/\binom{n-t}{k-t},
\end{equation}
and $\mathcal{B}^c$ denotes block complements \cite{Heng2023}.
The connection between linear codes and $t$-designs manifests through two fundamental mechanisms: incidence matrices of $t$-designs generate linear codes \cite{Ding2015}, while certain codes inherently contain $t$-designs. A canonical construction method in coding theory is summarized below \cite{Heng2023}. For an $[n,k]$ code $\mathcal{C}$ over $\mathbb{F}_q$ with coordinate positions $\mathcal{P}(\mathcal{C}) := \{1,\ldots,n\}$, define the scaled support multiset
\begin{equation}
  \mathcal{B}_w(\mathcal{C}) := \frac{1}{q-1}\{\!\!\{\operatorname{supp}(\mathbf{c}) : \mathbf{c} \in \mathcal{C},\ \operatorname{wt}(\mathbf{c}) = w\}\!\!\},
\end{equation}
where the scaling accounts for scalar multiples \cite{Tang2020}. When $(\mathcal{P}(\mathcal{C}), \mathcal{B}_w(\mathcal{C}))$ forms a $t$-$(n,w,\lambda)$ design, the parameters satisfy
\begin{equation}\label{eq(1)}
  b = \frac{A_w}{q-1}, \quad \lambda = \frac{\binom{w}{t}A_w}{(q-1)\binom{n}{t}}.
\end{equation}
The Assmus-Mattson Theorem \cite{Assmus1969} provides sufficient conditions for the support structure of a code to form $t$-designs. Utilizing this criterion, numerous linear codes admitting
$t$-designs have been discovered (see, for example, \cite{Ding2017,Ding2018,DingT2020,DingT2022,Heng20231,Li20231,Tang2019,WangX2023,Xu2022}), and some new $t$-designs have been obtained.

The investigation of NMDS codes supporting $t$-designs originated with Golay's 1949 discovery of the $[11,6,5]_3$ code admitting a $4$-design \cite{Golay1949}. A fundamental breakthrough came in 2020 when Ding and Tang \cite{Ding2020} constructed the first infinite families of NMDS codes with length $q+1$ and dimension $4$ over $\mathbb{F}_q$ ($q=2^m$ and $3^m$) supporting $2$-designs and $3$-designs, resolving a seventy-year open problem.
Heng et al. \cite{Heng2023,Heng20231} subsequently developed NMDS codes over $\mathbb{F}_{2^m}$ ($m\geq3$) with lengths $q-1$ or $q$ and fixed dimensions $3\leq k\leq6$ supporting $2$-designs or $3$-designs, while conjecturing similar results for arbitrary dimensions \cite[Conjecture 36]{Heng2023}.
Xu et al. constructed an infinite family of NMDS codes with length $q$ and dimension $3$ over $\mathbb{F}_{3^m}$ supporting $2$-designs \cite{Xu2022}, and NMDS codes with length $q+1$ and dimension $7$ over $\mathbb{F}_{3^{2m}}$ supporting $3$-designs \cite{Xu2024}. Furthermore, Tang and Ding \cite{Tang2021} presented the first infinite family of NMDS codes with length $q+1$ and dimension $6$ over $\mathbb{F}_{2^m}$ supporting $4$-designs. Up to now, there are only a few known infinite families of NMDS codes supporting $t$-designs for $t\geq2$ in the literature.

\subsection{Our motivations and contributions}
Through the analysis of the aforementioned literature, especially the papers \cite{DingY2024,Heng2023,Li2023,Xu2023}, our investigation is motivated by two fundamental questions as follows:
\begin{itemize}
\item [(i)] Is there an infinite family of NMDS codes with arbitrary dimensions supporting $t$-designs for $t\geq2$?
\item [(ii)] Is Heng-Wang's conjecture \cite{Heng2023} on flexible NMDS codes supporting 2-designs valid?
\end{itemize}

In this paper, we focus on the constructions of infinite families of NMDS codes with arbitrary dimensions by selecting some special matrices over finite fields as their generator matrices.
By using the subset-sum theorem provided by \cite{Li2008}, we construct four new infinite families of NMDS codes with arbitrary dimensions, two of which support 2-designs and 3-designs. Our key contributions are summarized as follows:
\begin{itemize}
\item We propose two new infinite families of NMDS codes with lengths $q+1,q+2$ and arbitrary dimensions over $\mathbb{F}_q$ for any prime power $q$ (See Theorems \ref{th3.1} and \ref{th3.4}).
 Notably, these codes generalize the related constructions of NMDS codes given in \cite{DingY2024}. Compared with the known ones constructed in \cite{Ding2020,Heng2022,Li2023,Wang2021}, these codes have different generator matrices, weight distributions and more flexible parameters.
\item We construct two infinite families of NMDS codes of lengths $2^m-1$ and $2^m$ supporting $2$-designs and $3$-designs with arbitrary dimensions over $\mathbb{F}_{2^m}$ (See Theorems \ref{th4.1} and \ref{th4.4}). To our knowledge, these represent the first infinite families of NMDS codes with arbitrary dimensions supporting $t$-designs for $t\geq2$. Our results not only generalize fixed-dimension constructions in \cite{Heng2023,Heng20231,Xu2022}, but also provide an affirmative resolution to the Heng-Wang conjecture \cite{Heng2023}.
\item We provide an alternative proof of the well-known subset sum problem for binary case. Some new recurrence relations for  $N(k,0,{\bF_{2^m}^*})$ and $N(k,0,{\bF_{2^m}})$ defined by Eq.(\ref{eq(2)}) and new combinatorial identity are derived (see Appendix).
\end{itemize}

Furthermore, some known NMDS codes supporting $2$-designs and $3$-designs are listed in Tables \ref{tab1} and \ref{tab2}, respectively. Compared with these codes, it is easy to see that our NMDS codes supporting $2$-designs and $3$-designs constructed in this paper have more flexible parameters.

\begin{table}[h!]
{ \small
\centering
\setlength{\tabcolsep}{0.8mm}
\caption{Some known NMDS codes supporting $2$-designs}
\label{tab1}
 \begin{tabular}{llllc}
\toprule
Parameters($\mathcal{C}/\mathcal{C}^\bot$)&$A_{n-k}=A_{k}^\bot$ & $t$-$(n,n-k,\lambda)$/ $t$-$(n,k,\lambda^\bot)$ & Constraints & Refs.\\
\midrule
\tabincell{l}{$[q-1,3,q-4]_q$\\$[q-1,q-4,3]_q$}&$\frac{(q-1)^2(q-2)}{6}$&\tabincell{l}{2-$(q-1,q-4,\frac{(q-4)(q-5)}{6})$\\2-$(q-1,4,\frac{(q-4)(q-7)}{2})$}&$q=2^m,m\geq3$
&\cite{Heng2023}\\
\midrule
\tabincell{l}{$[q-1,4,q-5]_q$\\$[q-1,q-5,4]_q$}&$\frac{(q-1)^2(q-2)(q-8)}{24}$&\tabincell{l}{2-$(q-1,q-5,\frac{(q-5)(q-6)(q-8)}{24})$\\2-$(q-1,4,\frac{(q-8)}{2})$}&\tabincell{c}{
$q=2^m,m>3$\\ is odd}&\cite{Heng2023}\\
\midrule
\tabincell{l}{$[q-1,4,q-5]_q$\\$[q-1,q-5,4]_q$}&$\frac{(q-1)^2(q-2)(q-4)}{24}$&\tabincell{l}{2-$(q-1,q-5,\frac{(q-4)(q-5)(q-6)}{24})$\\2-$(q-1,4,\frac{(q-4)}{2})$}&$q=2^m,m\geq3$
&\cite{Heng2023}\\
\midrule
\tabincell{l}{$[q-1,5,q-6]_q$\\$[q-1,q-6,5]_q$}&$\frac{(q-1)^2(q-2)(q-5)(q-8)}{120}$&\tabincell{l}{2-$(q-1,q-6,\frac{(q-5)(q-6)(q-7)(q-8)}{120})$\\2-$(q-1,5,\frac{(q-5)(q-8)}{4})$}
&\tabincell{c}{$q=2^m,m>3$\\ is odd}&\cite{Heng2023}\\
\midrule
\tabincell{l}{$[q-1,5,q-6]_q$\\$[q-1,q-6,5]_q$}&$\frac{(q-1)^2(q-2)(q-4)(q-8)}{120}$&\tabincell{l}{2-$(q-1,q-6,\frac{(q-4)(q-6)(q-7)(q-8)}{120})$\\2-$(q-1,5,\frac{(q-4)(q-8)}{4})$}
&$q=2^m,m>3$&\cite{Heng2023}\\
\midrule
\tabincell{l}{$[q-1,6,q-7]_q$\\$[q-1,q-7,6]_q$}&$\frac{(q-1)^2(q-2)(q-4)(q-6)(q-8)}{720}$&\tabincell{l}{2-$(q-1,q-7,\frac{(q-4)(q-6)(q-7)(q-8)^2}{720})$\\2-$(q-1,6,\frac{(q-4)(q-6)(q-8)}{24})$}
&$q=2^m,m>3$&\cite{Heng2023}\\
\midrule
\tabincell{l}{$[q-1,6,q-7]_q$\\$[q-1,q-7,6]_q$}&$\frac{(q-1)^2(q-2)(q-5)(q-6)(q-8)}{720}$&\tabincell{l}{2-$(q-1,q-7,\frac{(q-5)(q-6)(q-7)(q-8)^2}{720})$\\2-$(q-1,6,\frac{(q-5)(q-8)}{6})$}
&\tabincell{c}{$q=2^m,m>3$\\ is odd}&\cite{Heng2023}\\
\midrule
\tabincell{l}{$[q,3,q-3]_q$\\$[q,q-3,3]_q$}&$\frac{q(q-1)^2}{6}$&\tabincell{l}{2-$(q,q-3,\frac{(q-3)(q-4)}{6})$\\2-$(q,3,1)$}
&\tabincell{c}{$q=3^m,m\geq3$}&\cite{Xu2022}\\
\midrule
\tabincell{l}{$[q+1,q-3,4]_q$\\$[q+1,4,q-3]_q$}&$\frac{(q-4)(q-1)q(q+1)}{24}$&\tabincell{l}{2-$(q+1,4,\frac{q-4}{2})$\\ 2-$(q+1,q-3,\frac{(q-3)(q-4)^2}{24})$}&
\tabincell{c}{$q=2^m,m\geq4$\\ is even}&\cite{Ding2020}\\
\midrule
\tabincell{l}{$[q-1,k,q-1-k]_q$\\$[q-1,q-1-k,k]_q$}&$(q-1)N(k,0,\mathbb{F}_q^*)$&\tabincell{l}{2-$(q-1,q-1-k,\lambda_1)$\\
2-$(q-1,k,\lambda_1^c)$}&\tabincell{l}{$q=2^m,m\geq3$\\$3\leq k\leq q-4$}&Th.\ref{th4.1}\\
\bottomrule
 \end{tabular}
}
\end{table}

\begin{table}[h!]
{\small
\centering
\setlength{\tabcolsep}{1.8mm}
\caption{Some known NMDS codes supporting $3$-designs}
\label{tab2}
 \begin{tabular}{llllc}
\toprule
Parameters($\mathcal{C}/\mathcal{C}^\bot$)&$A_{n-k}=A_{k}^\bot$&$t$-$(n,n-k,\lambda)$/ $t$-$(n,k,\lambda^\bot)$&Constraint&Refs.\\
\midrule
\tabincell{c}{$[q,4,q-4]_q$\\$[q,q-4,4]_q$}&$\frac{q(q-1)^2(q-2)}{24}$&\tabincell{l}{3-$(q,q-4,\frac{(q-4)(q-5)(q-6)}{24})$\\ 3-$(q,4,1)$}
&$q=2^m$,$m\geq2$&\cite{Heng20231,Xu2022}\\
\midrule
\tabincell{c}{$[q,5,q-5]_q$\\$[q,q-5,5]_q$}&$\frac{q(q-1)^2(q-2)(q-8)}{120}$&\tabincell{l}{3-$(q,q-5,\frac{(q-5)(q-6)(q-7)(q-8)}{120})$\\3-$(q,5,\frac{q-8}{2})$}&
\tabincell{c}{$q=2^m,m>3$\\ is odd}&\cite{Heng2023}\\
\midrule
\tabincell{c}{$[q,6,q-6]_q$\\$[q,q-6,6]_q$}&$\frac{q(q-1)^2(q-2)(q-4)(q-8)}{720}$&\tabincell{l}{3-$(q,q-6,\frac{(q-4)(q-6)(q-7)(q-8)^2}{720})$\\3-$(q,6,\frac{(q-4)(q-8)}{6})$}&$q=2^m,m>3$ &\cite{Heng2023}\\
\midrule
\tabincell{c}{$[q,6,q-6]_q$\\$[q,q-6,6]_q$}&$\frac{q(q-1)^2(q-2)(q-5)(q-8)}{720}$&\tabincell{l}{3-$(q,q-6,\frac{(q-5)(q-6)(q-7)(q-8)^2}{720})$\\3-$(q,6,\frac{(q-5)(q-8)}{6})$}
&\tabincell{c}{$q=2^m,m>3$ \\is odd}&\cite{Heng2023}\\
\midrule
\tabincell{c}{$[q+1,q-3,4]_q$\\$[q+1,4,q-3]_q$}&$\frac{(q-1)^2q(q+1)}{24}$&\tabincell{l}{3-$(q+1,4,1)$\\3-$(q+1,q-3,\frac{(q-3)(q-4)(q-5)}{24})$}&$q=3^m,m\geq2$&\cite{Ding2020,Xu2022}\\
\midrule
\tabincell{c}{$[q+1,7,q-6]_q$\\$[q+1,q-6,7]_q$}&&\tabincell{l}{3-$(q+1,q-6,\lambda_2\tbinom{q-2}{7}/\tbinom{q-2}{4})$\\ 3-$(q+1,7,\lambda_2)$}&
\tabincell{l}{$q=3^m,m\geq3$,\\ $\lambda_2$ are some\\positive integers} &\cite{Xu2024}\\
\midrule
\tabincell{c}{$[q,k,q-k]_q$\\$[q,q-k,k]_q$}&$(q-1)N(k,0,\mathbb{F}_q)$&\tabincell{l}{3-$(q,q-k,\lambda_2)$\\
3-$(q,k,\lambda_2^c)$}&\tabincell{c}{$q=2^m,m\geq3$,\\$4\leq k\leq q-4$\\is even}&Th.\ref{th4.4}\\
\bottomrule
 \end{tabular}
}
\end{table}

\subsection{Organization of this Paper}
This paper is arranged as follows. Section \ref{sec.2} introduces some basic properties and known results related to NMDS codes and subset sums over finite fields. Section \ref{sec.3} constructs two new infinite families of NMDS codes with arbitrary dimensions. Section \ref{sec.4} presents two infinite families of NMDS codes with arbitrary dimensions supporting $2$-designs and $3$-designs. Section \ref{sec.5} makes some conclusions of this paper. Appendix provides an alternative proof of subset sum problem for binary case.

\section{Preliminaries}\label{sec.2}
In this section, we introduce some basic properties and known results related to NMDS codes and subset sums over finite fields, which are utilized subsequently. Unless stated otherwise, from this point forward, we will use the notation outlined below:
\begin{itemize}
  \item $\mathbb{F}_q$ denotes the finite field with $q$ elements, where $q=p^m$ is a prime power, $m \geq 3$ if $p=2$ and $m \geq 2$ otherwise.
  \item $\mathbb{F}_q^* =\mathbb{F}_q\setminus\{0\}$ is enumerated as $\{\theta_1, \theta_2, \ldots, \theta_{q-1}\}$ with pairwise distinct elements.
  \item $\mathcal{C}$ is an $[n,k,d]_q$ linear code, and $\mathcal{C}^\perp$ is the dual code of $\mathcal{C}$.
  \item $d({\C})$ (resp. $d({\C}^\perp$)) denotes the minimum Hamming distance of $\mathcal{C}$ (resp. $\mathcal{C}^\perp$).
  \item $A_i$ (resp.\ $A_i^\perp$) denotes the number of codewords with Hamming weight $i$ in $\mathcal{C}$ (resp.\ $\mathcal{C}^\perp$).
  \item $\#S$ denotes the cardinality of finite set $S$.
  \item $\mathbf{0}_n$ signifies the zero column vector in $\mathbb{F}_q^n$.
  \item $\det(M)$ and $\rank(G)$ denote the determinant of square matrix $M$ and the rank of matrix $G$, respectively.
  \item $\binom{n}{m}$ denotes the number of ways to choose $m$ distinct items from $n$ (where $0 \leq m \leq n$) without regard to order, calculated as $\frac{n!}{m!(n-m)!}$.
\end{itemize}
\subsection{NMDS codes}
For an $[n,k,n-k]_q$ NMDS code ${\C}$, let $(A_0,A_1,\ldots,A_n)$ and $(A_0^\perp,A_1^\perp,\ldots,A_n^\perp)$ denote the weight distributions of $\mathcal{C}$ and $\mathcal{C}^\perp$, respectively. We have the following known results.
\begin{lemma}\cite{DingT2022,Dodunekov1995}\label{le2.1}
Let $\mathcal{C}$ be an $[n,k,n-k]_q$ NMDS code. Then the weight distributions of $\mathcal{C}$ and $\mathcal{C}^\perp$ are given by
$$A_{k+t}^\perp=\binom{n}{k+t}\sum_{j=0}^{t-1}(-1)^j\binom{k+t}{j}(q^{t-j}-1)+(-1)^t\binom{n-k}{t}A_{k}^\perp$$
for $t\in\{1,2,\ldots,n-k\}$, and
$$A_{n-k+t}=\binom{n}{k-t}\sum_{j=0}^{t-1}(-1)^j\binom{n-k+t}{j}(q^{t-j}-1)+(-1)^t\binom{k}{t}A_{n-k}$$
for $t\in\{1,2,\ldots,k\}$.
\end{lemma}

The following lemma establishes a natural correspondence between the minimum weight codewords of an NMDS code and those of its dual code.
\begin{lemma}\cite{DingT2022}\label{le2.2}
Let $\mathcal{C}$ be an $[n,k,n-k]_q$ NMDS code. Then for every minimum weight codeword $\mathbf{c}$ in $\mathcal{C}$, there exists, up to a multiple, a unique minimum weight codeword $\mathbf{c}^\perp$ in $\mathcal{C}^\perp$ such that ${\rm supp}(\mathbf{c})\cap {\rm supp}(\mathbf{c}^\perp)=\emptyset$. In particular, $\mathcal{C}$ and $\mathcal{C}^\perp$ have the same number of minimum weight codewords.
\end{lemma}
\subsection{The subset sum problem over finite fields}
Let $D \subseteq \mathbb{F}_q$ be an $n$-element subset. The \emph{subset sum problem} over $\mathbb{F}_q$ asks whether there exists a $k$-subset $\{x_1,\ldots, x_k\} \subseteq D$ such that $\sum_{i=1}^k x_i = b$ for given $b\in \mathbb{F}_q$. A natural extension involves determining the exact number of such $k$-subsets in finite fields. Specifically, we define the counting function
\begin{equation}\label{eq(2)}
N(k,b,D) = \#\left\{\{x_1, x_2, \ldots, x_k\} \subseteq D : \sum_{i=1}^k x_i = b \right\}.
\end{equation}
By means of a purely combinatorial method, Li and Wan \cite{Li2008} derived an explicit formula for $N(k,b,D)$ in two cases: $D=\mathbb{F}_q$ and $D =  \mathbb{F}_q^*$.

\begin{lemma}\cite{Li2008}\label{le2.5}
Let $\bF_q$ be a finite field of characteristic $p$, and let $b$ be a given element of ${\bF_q}$. Define $v(b)=-1$ if $b\neq0$, and $v(b)=q-1$ if $b=0$. If $1\leq k\leq q$, then
\begin{equation*}
\begin{aligned}
N(k,b,{\bF_q})=
\left\{
  \begin{array}{ll}
   \frac{1}{q}\binom{q}{k}, & \hbox{if $k$ is not multiple of $p$;} \\
    \frac{1}{q}\binom{q}{k}+(-1)^{k+k/p}\frac{v(b)}{q}\binom{q/p}{k/p}, & \hbox{if $k$ is multiple of $p$.}
  \end{array}
\right.
\end{aligned}
\end{equation*}
Furthermore, if $1\leq k\leq q-1$, then
\begin{equation*}
\begin{aligned}
N(k,b,{\bF_q^*})=\frac{1}{q}\binom{q-1}{k}+(-1)^{k+\lfloor k/p\rfloor}\frac{v(b)}{q}\binom{q/p-1}{\lfloor k/p\rfloor},
\end{aligned}
\end{equation*}
where $\lfloor\cdot\rfloor$ is the floor function.
\end{lemma}

In 2021, Pavone \cite{Pavone2021} provided an alternative proof of the above two formulas by using a different combinatorial approach. In this paper's Appendix, we also present an alternative algebraic proof of $N(k,0,\mathbb{F}_{2^m}^*)$ and $N(k,0,\mathbb{F}_{2^m})$, and derive some new explicit formulas equivalent to those in Lemma \ref{le2.5} for binary case. Moreover, some new recurrence relations about $N(k,0,\mathbb{F}_{2^m}^*)$ and a combinatorial identity are obtained. For more results on subset sum problem, see \cite{Etzion1994,Falcone2021,Kosters2013,Li2012}.

The following lemma establishes an explicit formula for computing the generalized Vandermonde determinant with one omitted row.
\begin{lemma}\cite{Heng2023,Heineman1929}\label{le2.6}
Let $0\leq i\leq k$ be an integer. Then the following determinant identity holds:
\begin{equation}\label{Vandermonde}
\det\left(
\begin{matrix}
1 & 1 & \cdots & 1 \\
x_1 & x_2 & \cdots & x_k \\
\vdots & \vdots & \ddots & \vdots \\
x_1^{i-1} & x_2^{i-1} & \cdots & x_k^{i-1} \\
x_1^{i+1} & x_2^{i+1} & \cdots & x_k^{i+1} \\
\vdots & \vdots & \ddots & \vdots \\
x_1^k & x_2^k & \cdots & x_k^k
\end{matrix}
\right)
= \left(\prod_{1\leq m<j\leq k}(x_j-x_m)\right)\sum_{\substack{1\leq j_1<\cdots<j_{k-i}\leq k}}x_{j_1}\cdots x_{j_{k-i}}.
\end{equation}
\end{lemma}

In Lemma \ref{le2.6}, if $\{x_1,x_2,\ldots,x_k\}\subseteq D$ is a $k$-subset of $\mathbb{F}_q$ with $i=k-1$, then the number of solutions for which determinant (5) vanishes precisely equals $N(k,0,D)$ defined in Eq.(\ref{eq(2)}).

\section{New infinite families of NMDS codes with arbitrary dimensions}\label{sec.3}
In this section, we construct two new infinite families of NMDS codes with arbitrary dimensions over ${\bF_q}$ for any prime power $q$.
\subsection{The first infinite family of NMDS codes with parameters $[q+1,k,q+1-k]_q$}
Let ${\C_1}$ denote the linear code over ${\bF_q}$ with generator matrix
\begin{equation}\label{G1}
\begin{aligned}
G_1 = \begin{pmatrix}
  \theta_1^{k-1}  & \cdots & \theta_{q-1}^{k-1} & 1 & 0 \\
  \theta_1^{k-2}  & \cdots & \theta_{q-1}^{k-2} & 0 & 1 \\
  \theta_1^{k-3}  & \cdots & \theta_{q-1}^{k-3} & 0 & 0 \\
  \vdots                 & \ddots & \vdots              & \vdots & \vdots \\
  \theta_1              & \cdots & \theta_{q-1}       & 0 & 0 \\
  1                            & \cdots & 1                  & 0 & 0
\end{pmatrix},
\end{aligned}
\end{equation}
where $4\leq k\leq q-3$ if $p=2$ and $3\leq k\leq q-2$ if $p\neq2$.
In the following, we prove that ${\C_1}$ is an NMDS code with parameters $[q+1,k,q+1-k]_q$ and determine the weight distribution of this code. We first consider the parameters of the dual of ${\C_1}$.

\begin{lemma}\label{le5}
Let $\mathcal{C}_1$ be a linear code over ${\bF_q}$ with generator matrix $G_1$ defined in (\ref{G1}). Then $d(\mathcal{C}_1^\perp) \geq k$.
\end{lemma}

\begin{proof}
By Corollary 1.4.14 of \cite{Huffman2003}, it suffices to verify that any $k-1$ columns of $G_1$ are linearly independent over $\mathbb{F}_q$. We prove the case involving $k-3$ columns from the first $q-1$ columns and both last two columns of $G_1$. Other cases follow similar arguments.

Consider columns $\{\theta_{i_1}, \ldots, \theta_{i_{k-3}}\}$ from the first $q-1$ columns, along with the $q$-th and $(q+1)$-th columns of $G_1$. The corresponding $k \times (k-1)$ submatrix is
\[
G_{1,1} = \begin{pmatrix}
    \theta_{i_1}^{k-1}  & \cdots & \theta_{i_{k-3}}^{k-1} & 1 & 0 \\
    \theta_{i_1}^{k-2}  & \cdots & \theta_{i_{k-3}}^{k-2} & 0 & 1 \\
    \vdots & \ddots & \vdots & \vdots & \vdots \\
    \theta_{i_1} & \cdots & \theta_{i_{k-3}} & 0 & 0 \\
    1 & \cdots & 1 & 0 & 0
\end{pmatrix}.
\]
Removing the last row gives the determinant
\[
\det\begin{pmatrix}
    \theta_{i_1}^{k-1} & \cdots & \theta_{i_{k-3}}^{k-1} & 1 & 0 \\
    \theta_{i_1}^{k-2} & \cdots & \theta_{i_{k-3}}^{k-2} & 0 & 1 \\
    \vdots & \ddots & \vdots & \vdots & \vdots \\
    \theta_{i_1} & \cdots & \theta_{i_{k-3}} & 0 & 0
\end{pmatrix} = (-1)^{\frac{(k-3)(k-4)}{2}} \prod_{j=1}^{k-3}\theta_{i_j} \cdot \prod_{1\leq a<b\leq k-3}(\theta_{i_b} - \theta_{i_a}).
\]
Since $\theta_{i_j}$ are distinct non-zero elements, the determinant is non-zero in $\mathbb{F}_q$. It is easy to check that $G_{1,1}$ has no subdeterminant of order $k$. It follows that any $k-3$ columns from the first $q-1$ columns of $G_1$ together with the $q$-th and $(q+1)$-th columns are ${\bF_q}$-linearly independent.
\end{proof}

\begin{lemma}\label{le6}
Let $\mathcal{C}_1$ be a linear code over ${\bF_q}$ with generator matrix $G_1$ defined in (\ref{G1}). Then $d(\mathcal{C}_1^{\perp})= k$ and the number of its minimum weight codewords is
$$
 A_k^\perp =\frac{q-1}{q}\binom{q-1}{k-1}+(-1)^{k-1+\lfloor\frac{k-1}{p}\rfloor}\frac{(q-1)^2}{q}\binom{\frac{q}{p}-1}{\lfloor\frac{k-1}{p}\rfloor}.
$$
\end{lemma}
\begin{proof}
We first prove that $d(\mathcal{C}_1^{\perp})= k$. By Lemma \ref{le5}, we only need to show that there exist codewords $\overline{\mathbf{c}}\in\mathcal{C}_1^\perp$ with Hamming weight $k$. Next, we present the following analysis.

Assume that the $k$ nonzero coordinates of codewords $\bar{\mathbf{c}}\in\mathcal{C}_1^\bot$ are all at the first $q-1$ locations, i.e.,
$$\bar{\mathbf{c}}=(\underbrace{0,\ldots,0,c_{i_1},0,\ldots,0,c_{i_2},\ldots,0,\ldots,0,c_{i_{k-1}},0,\ldots,0,c_{i_k},0,\ldots}_{q-1},0,0),$$
then $\bar{\mathbf{c}}\cdot G_1^\top=\mathbf{0}_{k}$. Precisely, we have the following homogeneous system of equations holds:
\begin{equation}\label{h-system1}
\begin{aligned}
G_{2,1}\cdot(c_{i_1},c_{i_2},\ldots,c_{i_{k-1}},c_{i_k})^\top=\mathbf{0}_k,
\end{aligned}
\end{equation}
where
$$
G_{2,1}=
\begin{pmatrix}
    \theta_{i_1}^{k-1}  & \cdots & \theta_{i_k}^{k-1}  \\
    \vdots&\ddots&\vdots\\
    \theta_{i_1}  & \cdots & \theta_{i_k}  \\
    1&\cdots&1\\
 \end{pmatrix}.
$$
Then $\rank(G_{2,1})=k$ follows that
$$\det(G_{2,1})=(-1)^{\frac{k(k-1)}{2}}\prod_{1\leq b<a\leq k}(\theta_{i_a}-\theta_{i_b})\neq0.$$
This implies that the homogeneous system of equations (\ref{h-system1}) has and only has zero solutions, which is a contradiction.
Therefore,  $\mathcal{C}_1^\bot$ contains no codewords of weight $k$ with all nonzero coordinates in the first $q-1$ positions. By analogous arguments, $\mathcal{C}_1^\bot$ also lacks codewords of weight $k$, where
\begin{itemize}
\item $k-1$ nonzero coordinates lie in the first $q-1$ positions and one at the $q$-th position;
\item $k-2$ nonzero coordinates are in the first $q-1$ positions, with the last two at the
$q$-th and $(q+1)$-th positions.
\end{itemize}

Building upon the preceding analysis and Lemma \ref{le5}, we claim that there exist codewords with Hamming weight $k$ in $\mathcal{C}_1^\perp$ whose first $k-1$ non-zero coordinates lie within the first $q-1$ positions, with the remaining non-zero coordinate positioned at the $(q+1)$-th location. Indeed, let
$$\bar{\mathbf{c}}'=(\underbrace{0,\ldots,0,c_{i_1},0,\ldots,0,c_{i_2},\ldots,0,\ldots,0,c_{i_{k-1}},0,\ldots,0}_{q-1},0,c_{i_k}),$$
then $\bar{\mathbf{c}}'\cdot G_1^\top=\mathbf{0}_{k}$. Precisely, we have the following homogeneous system of equations holds:
\begin{equation*}
\begin{aligned}
G_{2,2}\cdot(c_{i_1},c_{i_2},\ldots,c_{i_{k-1}},c_{i_k})^\top=\mathbf{0}_k,
\end{aligned}
\end{equation*}
where
$$G_{2,2}=
\begin{pmatrix}
    \theta_{i_1}^{k-1} & \cdots & \theta_{i_{k-1}}^{k-1}& 0  \\
    \theta_{i_1}^{k-2}  & \cdots & \theta_{i_{k-1}}^{k-2}& 1  \\
    \vdots&\ddots&\vdots&\vdots\\
    \theta_{i_1}  & \cdots & \theta_{i_{k-1}}& 0  \\
    1&\cdots&1&0\\
\end{pmatrix}.
$$
By Lemma \ref{le2.6}, we have
$$\det(G_{2,2}) = (-1)^{\frac{k(k-1)}{2} + 1} \sum_{j=1}^{k-1} \theta_{i_j} \cdot \prod_{\substack{1 \leq a < b \leq k-1}} (\theta_{i_b} - \theta_{i_a}).$$
Then $\det(G_{2,2})=0$ if and only if
$\sum_{j=1}^{k-1}\theta_{i_j}=0.$
By Lemma~\ref{le2.5}, the number of distinct $(k-1)$-subsets $\{\theta_{i_1}, \ldots, \theta_{i_{k-1}}\} \subseteq \mathbb{F}_q^*$ satisfying $\det(G_{2,2}) = 0$ equals $N(k-1, 0, \mathbb{F}_q^*)$, which directly implies $\rank(G_{2,2}) = k-1$. Notably, Lemma~\ref{le2.5} establishes that $N(k-1, 0, \mathbb{F}_q^*) = 0$ if and only if either: (i)~$k \in \{3, q-2, q-1\}$ for $p=2$, or (ii)~$k \in \{2, q-1\}$ for $p \neq 2$.

Consequently, there exist codewords in $\mathcal{C}_1^\perp$ with Hamming weight $k$. Combining with Lemma \ref{le5},
we can deduce that $d(\mathcal{C}_1^\perp)=k$. By Lemma~\ref{le2.5}, the total number of codewords with minimum weight $k$ in $\mathcal{C}_1^\perp$ is
\begin{equation*}
A_k^{\bot}=(q-1)N(k-1,0,{\bF_q^*})
=\frac{q-1}{q}\binom{q-1}{k-1}+(-1)^{k-1+\lfloor\frac{k-1}{p}\rfloor}\frac{(q-1)^2}{q}\binom{\frac{q}{p}-1}{\lfloor\frac{k-1}{p}\rfloor},\\
\end{equation*}
where $4\leq k\leq q-3$ if $p=2$ and $3\leq k\leq q-2$ if $p\neq2$.
\end{proof}
Based on the results of Lemmas \ref{le5} and \ref{le6}, we have the following theorem.
\begin{theorem}\label{th3.1}
Let $\mathcal{C}_1$ be a linear code over ${\bF_q}$ with generator matrix $G_1$ defined in (\ref{G1}). Then $\mathcal{C}_1$ is an NMDS code with parameters $[q+1,k,q+1-k]_q$, has exactly
$$
A_{q+1-k} = \frac{q-1}{q}\binom{q-1}{k-1}+(-1)^{k-1+\lfloor\frac{k-1}{p}\rfloor}\frac{(q-1)^2}{q}\binom{\frac{q}{p}-1}{\lfloor\frac{k-1}{p}\rfloor}
$$
minimum weight codewords, and its weight distribution is determined by Lemma \ref{le2.1}.
\end{theorem}
\begin{proof}
Since $\theta_{i_b}\neq\theta_{i_a}$ for any $1 \leq i_a\neq i_b\leq k$, we know that the determinant of the submatrix of $G_1$ formed by selecting \( k \) columns \( \theta_{i_1}, \ldots, \theta_{i_k} \) is
\[
\det\left(\begin{matrix}
    \theta_{i_1}^{k-1} & \cdots & \theta_{i_k}^{k-1} \\
    \vdots & \ddots & \vdots \\
    \theta_{i_1} & \cdots & \theta_{i_k} \\
    1 & \cdots & 1
\end{matrix}\right)
= (-1)^{\frac{k(k-1)}{2}} \prod_{1 \leq {i_a} < {i_b} \leq k} (\theta_{i_b} - \theta_{i_a}) \neq 0.
\]
Then $\rank(G_1)=k$. Hence, ${\rm dim}(\mathcal{C}_1)=\rank(G_1)=k$ follows that $G_1$ has no subdeterminant of order $k+1$.  This implies that $\dim{(\mathcal{C}_1^\perp)}=q+1-\dim{(\mathcal{C}_1)}=q+1-k$. By Lemma \ref{le6}, $\mathcal{C}_1^\perp$ is an AMDS code with parameters $[q+1,q+1-k,k]_q$.

We now prove that $d(\mathcal{C}_1)=q+1-k$. Assume towards contradiction that $d(\mathcal{C}_1)\leq q-k$. Let $\mathbf{c}=a_1\mathbf{g}_1+\cdots+a_k\mathbf{g}_k$ be a minimum weight nonzero codeword, where $\mathbf{g}_i$ is the $i$-th row of $G_1$. By assumption, $\mathbf{c}$ has at least $k+1$ zero coordinates. We analyze three cases:

\begin{itemize}
  \item \noindent {\bf Case 1:}~Suppose $\mathbf{c}$ has $k-1$ zeros among the first $q-1$ coordinates and two zeros in the final positions. This implies there exist $k-1$ distinct $\theta_{i_1},\ldots,\theta_{i_{k-1}}\in\mathbb{F}_q^*$ satisfying:
\begin{equation*}
\begin{cases}
a_1\theta_i^{k-1}+\cdots+a_{k-1}\theta_i + a_k = 0, & \forall i\in\{i_1,\ldots,i_{k-1}\}, \\
a_1 = a_2 = 0.
\end{cases}
\end{equation*}
The polynomial $f(x)=a_3x^{k-3}+\cdots+a_k$ must possess $k-1$ roots, forcing $a_3=\cdots=a_k=0$. Thus $\mathbf{c}=\mathbf{0}$, which is a contradiction.
  \item \noindent {\bf Case 2:}~Suppose $\mathbf{c}$ has $k$ zeros among the first $q-1$ coordinates and one zero in the final positions. This gives either:
\begin{equation*}
\begin{cases}
a_1\theta_i^{k-1}+\cdots+a_k = 0, & \forall i\in\{i_1,\ldots,i_k\}, \\
a_1 = 0,
\end{cases}
\text{or}
\begin{cases}
a_1\theta_i^{k-1}+\cdots+a_k = 0, &\forall i\in\{i_1,\ldots,i_k\}, \\
a_2 = 0.
\end{cases}
\end{equation*}
In both subcases, the resulting polynomial (degree $\leq k-2$ or $\leq k-1$) must possess $k$ roots, forcing all coefficients zero. Hence, $\mathbf{c}=\mathbf{0}$, which is a contradiction.
  \item \noindent {\bf Case 3:}~Suppose $\mathbf{c}$ has $k+1$ zeros among the first $q-1$ coordinates. The system:
\begin{equation*}
a_1\theta_i^{k-1}+\cdots+a_k = 0, \quad \forall i\in\{i_1,\ldots,i_{k+1}\},
\end{equation*}
implies the polynomial of degree $\leq k-1$ must possess $k+1$ roots. Thus, $a_1=\cdots=a_k=0$, and so, $\mathbf{c}=\mathbf{0}$. This is a contradiction.
\end{itemize}

All cases lead to contradictions, therefore $d(\mathcal{C}_1)\geq q-k+1$. By the Singleton bound, $d(\mathcal{C}_1) \leq q - k + 2$. If the equality holds, then $\mathcal{C}_1$ would be MDS, implying its dual $\mathcal{C}_1^\perp$ must also be MDS. However, $\mathcal{C}_1^\perp$ is known to be a $[q+1, q+1-k, k]_q$ AMDS code, which contradicts the MDS property of $\mathcal{C}_1^\perp$. Therefore, $d(\mathcal{C}_1) = q - k + 1$, and $\mathcal{C}_1$ is an NMDS code with parameters $[q+1, k, q - k + 1]_q$, where $4\leq k\leq q-3$ if $p=2$ and $3\leq k\leq q-2$ if $p\neq2$. By Lemmas~\ref{le2.2} and~\ref{le6}, the number of minimum weight codewords in ${\C}_1$ is
$$
A_{q+1-k}=A_k^\perp= \frac{q-1}{q}\binom{q-1}{k-1}+(-1)^{k-1+\lfloor\frac{k-1}{p}\rfloor}\frac{(q-1)^2}{q}\binom{\frac{q}{p}-1}{\lfloor\frac{k-1}{p}\rfloor},
$$
and the weight distribution of $\mathcal{C}_1$ follows from Lemma~\ref{le2.1}.
\end{proof}

Below, we present an example that support the result given in Theorem \ref{th3.1}, which has been verified using Magma program.
\begin{example}
Let ${\C_1}$ be a linear code over $\mathbb{F}_{3^m}$ with generator matrix
$$G_1=
\begin{pmatrix}
    \theta_1^4 &  \cdots & \theta_{q-1}^4 & 1 & 0 \\
    \theta_1^3 &  \cdots & \theta_{q-1}^3 & 0 & 1 \\
    \theta_1^2 &  \cdots & \theta_{q-1}^2 & 0 & 0 \\
    \theta_1   &  \cdots & \theta_{q-1}   & 0 & 0 \\
             1 &          \cdots & 1              & 0 & 0  \\
\end{pmatrix}.
$$
By Theorem \ref{th3.1}, $\mathcal{C}_{1}$ is a $[q+1,5,q-4]_q$ NMDS code. By Lemmas \ref{le2.5} and \ref{le6},
$$A_{q-4}=A_5^\perp=(q-1)N(4,0,{\bF_q^*})=\frac{(q-1)^2(q-3)(q-6)}{24}.$$ By Lemma \ref{le2.1}, the weight enumerator of $\mathcal{C}_1$ is
\begin{equation*}
\begin{aligned}
&1+\frac{(q-1)^2(q-3)(q-6)}{24}x^{q-4}+\frac{(q-1)^2(q^3-6q^2+43q-90)}{24}x^{q-3}\\
&~~+\frac{(q-1)^2(11q^2-39q+90)}{12}x^{q-2}+\frac{(q-1)(3q^4-5q^3+65q^2-117q+90)}{12}x^{q-1}\\
&~~+\frac{(q-1)(8q^4+25q^3-58q^2+139q-66)}{24}x^{q}+\frac{(q-1)^2(9q^3+2q^2+15q-18)}{24}x^{q+1}.
\end{aligned}
\end{equation*}

Let $m=2$ and $\omega$ be a primitive element of $\mathbb{F}_{3^2}$. Then
$$G_{1}=
\begin{pmatrix}
    1 & 2 & 1 & 2 & 1 & 2 & 1 & 2 & 1 & 0 \\
    1 & \omega^3 & \omega^6 & \omega & 2 & \omega^7 & \omega^2 & \omega^5 & 0 & 1 \\
    1 & \omega^2 & 2 & \omega^6 & 1 & \omega^2 & 2 & \omega^6 & 0 & 0 \\
    1 & \omega & \omega^2 & \omega^3 & \omega^4 & \omega^5 & \omega^6 & \omega^7 & 0 & 0 \\
    1 & 1 & 1 & 1 & 1 & 1 & 1 & 1 & 0 & 0 \\
 \end{pmatrix}.
$$
By Magma, $\mathcal{C}_{1}$ generated by $G_{1}$ is a $[10,5,5]_{3^2}$ NMDS code with weight enumerator
$$A(x)=1+48x^5+1440x^6+3360x^7+13560x^8+22400x^9+18240x^{10},$$
which is consistent with the result of Theorem \ref{th3.1}.
\end{example}
According to the results of Lemma \ref{le6} and Theorem \ref{th3.1}, we have the following result.
\begin{corollary}\label{cor:mds-condition}
Let $\mathcal{C}_1$ be a linear code over ${\bF_q}$ with generator matrix $G_1$ defined in (\ref{G1}). Then $\mathcal{C}_1$ is an MDS code with parameters $[q+1,k,q+2-k]_q$ if either
\begin{itemize}
\item[1)] $p=2$ with $m \geq 3$ and $k \in \{3, q-2, q-1\}$; or
\item[2)] $p \neq 2$ with $m \geq 2$ and $k \in \{2, q-1\}$.
\end{itemize}
\end{corollary}

\begin{proof}
By Theorem \ref{th3.1}, we have $\dim(\mathcal{C}_1) = k$ and $\dim(\mathcal{C}_1^\perp) = q+1-k$. Through methodologies parallel to Lemmas \ref{le5} and \ref{le6}, we establish that any $k$ columns of $G_1$ are $\mathbb{F}_q$-linearly independent precisely under the stated conditions. This implies $d(\mathcal{C}_1^\perp) \geq k+1$.
The Singleton bound gives $d(\mathcal{C}_1^\perp) \leq k+1$. Therefore $\mathcal{C}_1^\perp$ has parameters $[q+1,q+1-k,k+1]_q$, making it MDS. Since duals of MDS codes remain MDS, the conclusion follows.
\end{proof}

\subsection{The second infinite family of NMDS codes with parameters $[q+2,k,q+2-k]_q$}
Let ${\C_2}$ denote the linear code over ${\bF_q}$ with generator matrix
\begin{equation}\label{G2}
\begin{aligned}
G_{2}=
\begin{pmatrix}
    \theta_1^{k-1}  & \cdots & \theta_{q-1}^{k-1}&0 & 1 & 0 \\
    \theta_1^{k-2}  & \cdots & \theta_{q-1}^{k-2}&0 & 0 & 1 \\
    \theta_1^{k-3}  & \cdots & \theta_{q-1}^{k-3}&0 & 0 & 0 \\
    \vdots&\ddots&\vdots&\vdots&\vdots&\vdots\\
    \theta_1      & \cdots & \theta_{q-1}              &0 & 0 & 0 \\
             1   & \cdots & 1                         &1 & 0 & 0 \\
 \end{pmatrix},
\end{aligned}
\end{equation}
where $4\leq k\leq q-2$ if $p=2$ and $3\leq k\leq q$ if $p\neq2$. By employing the proof method similar to that used in Theorem \ref{th3.1}, we can derive the second infinite family of NMDS codes with parameters $[q+2,k,q+2-k]_q$.
\begin{theorem}\label{th3.4}
Let $\mathcal{C}_2$ be a linear code over ${\bF_q}$ with generator matrix $G_2$ defined in (\ref{G2}). Then $\mathcal{C}_2$ is an NMDS code with parameters $[q+2,k,q+2-k]_q$, has exactly
\begin{equation*}
\begin{aligned}
A_{q+2-k}=\left\{
                       \begin{array}{ll}
                         \frac{q-1}{q}\binom{q}{k-1}, & \hbox{if $k-1$ is not multiple of $p$;} \\
                         \frac{q-1}{q}\binom{q}{k-1}+(-1)^{\frac{(k-1)(p+1)}{p}}\frac{(q-1)^2}{q}\binom{\frac{q}{p}}{\frac{k-1}{p}}, & \hbox{if $k-1$ is multiple of $p$}
                       \end{array}
                     \right.
\end{aligned}
\end{equation*}
minimum weight codewords, and its weight distribution is determined by Lemma \ref{le2.1}.
\end{theorem}

Below, we present an example that support the result given in Theorem \ref{th3.4}, which has been verified using Magma program.
\begin{example}
Let ${\C_2}$ be a linear code over ${\bF_{3^m}}$ with generator matrix
$$
G_{2}=
\begin{pmatrix}
    \theta_1^4 & \cdots & \theta_{q-1}^4 &0& 1 & 0 \\
    \theta_1^3  & \cdots & \theta_{q-1}^3 &0& 0 & 1 \\
    \theta_1^2  & \cdots & \theta_{q-1}^2 &0& 0 & 0 \\
    \theta_1      & \cdots & \theta_{q-1}   &0& 0 & 0 \\
             1   & \cdots & 1              &1& 0 & 0  \\
\end{pmatrix}.
$$
By Theorem \ref{th3.4},  $\mathcal{C}_{2}$ is a $[q+2,5,q-3]_q$ NMDS code. By Lemmas \ref{le2.5} and \ref{le6},
$$A_{q-3}=A_5^\perp=(q-1)N(4,0,{\bF_q})=\frac{(q-1)^2(q-2)(q-3)}{24}.$$
By Lemma \ref{le2.1}, the weight enumerator of $\mathcal{C}_2$ is
\begin{equation*}
\begin{aligned}
&1+\frac{(q-1)^2(q-2)(q-3)}{24}x^{q-3}+\frac{(q-1)^2(q^3-2q^2+27q-30)}{24}x^{q-2}\\
&+\frac{(q-1)(3q^3-6q^2+21q-10)}{4}x^{q-1}+\frac{(q-1)(3q^4+q^3+33q^2-43q+42)}{12}x^{q}\\
&+\frac{(q-1)^2(8q^3+29q^2-9q+30)}{24}x^{q+1}+\frac{(q-1)^3(3q^2+q+2)}{8}x^{q+2}.
\end{aligned}
\end{equation*}

Let $m=2$ and $\omega$ be a primitive element of $\mathbb{F}_{3^2}$. Then
$$G_2=
\left(
  \begin{array}{ccccccccccc}
    1 & 2 & 1 & 2 & 1 & 2 & 1 & 2 & 0& 1 & 0 \\
    1 & \omega^3 & \omega^6 & \omega & 2 & \omega^7 & \omega^2 & \omega^5 & 0& 0 & 1 \\
    1 & \omega^2 & 2 & \omega^6 & 1 & \omega^2 & 2 & \omega^6 & 0& 0 & 0 \\
    1 & \omega & \omega^2 & \omega^3 & \omega^4 & \omega^5 & \omega^6 & \omega^7 & 0& 0 & 0 \\
    1 & 1 & 1 & 1 & 1 & 1 & 1 & 1 & 1& 0 & 0 \\
  \end{array}
\right).
$$
By Magma, $\mathcal{C}_2$ generated by $G_2$ is an $[11,5,6]_{3^2}$ NMDS code with weight enumerator
$$A(x)=1+112x^6+2080x^7+3760x^8+15160x^9+21680x^{10}+16256x^{11},$$
which is consistent with the result of Theorem \ref{th3.4}.
\end{example}
Similar to Lemma \ref{cor:mds-condition}, we directly obtain the following result by Theorem \ref{th3.4}.
\begin{corollary}\label{coro12}
Let $\mathcal{C}_2$ be a linear code over ${\bF_q}$ with generator matrix $G_2$ defined in (\ref{G2}). Then $\mathcal{C}_2$ is an MDS code with parameters $[q+2,k,q+3-k]_q$ if $p=2$ with $m \geq 3$ and $k \in \{3, q-1\}$.
\end{corollary}

\begin{remark}
We compare the NMDS codes constructed in Theorems \ref{th3.1} and \ref{th3.4}  with the known ones in the literature as follows:
\begin{itemize}
  \item Setting $p=2,m\geq3$ and $k=4$ in Theorems \ref{th3.1} and \ref{th3.4}, we can directly obtain the corresponding results of Theorems 12 and 15 of \cite{DingY2024} with $h=1$, respectively.  This means that our results generalize the related constructions of NMDS codes given in \cite{DingY2024}.
  \item Setting $p=2,m\geq3$ and $k=4$ in Theorems \ref{th3.1} and \ref{th3.4}, our NMDS codes have different generator matrices and weight distributions compared with the NMDS codes with lengths $q+1,q+2$ and dimension $4$ over $\mathbb{F}_{2^m} (m\geq3~{\rm odd})$ provided in \cite{Heng2022}.
  \item Compared with the known NMDS codes with length $q+1$, dimension $4$ presented in \cite{Ding2020}, length $q+1$, dimension $3$ presented in \cite{Li2023} and lengths $q+1,q+2$, dimension $4$ presented in \cite{Wang2021}, our NMDS codes have different weight distributions and more flexible parameters.
\end{itemize}
\end{remark}
\section{New infinite families of NMDS codes with arbitrary dimensions supporting $t$-designs}\label{sec.4}
In this section, we construct two new infinite families of NMDS codes over ${\bF_{2^m}}$ with arbitrary dimensions supporting $2$-designs or $3$-designs.
\subsection{New NMDS codes with parameters $[2^m-1,k,2^m-1+k]_{2^m}$ supporting $2$-designs}
In this subsection, we prove that the Heng-Wang's Conjecture \cite{Heng2023} on NMDS codes holding $2$-designs.
Let ${\C_3}$ denote the third family of linear codes over ${\bF_{2^m}}$ with generator matrix
\begin{equation}\label{G3}
\begin{aligned}
G_3=
  \begin{pmatrix}
    1 &         1  & \cdots & 1               \\
    \theta_1   & \theta_2   & \cdots & \theta_{2^m-1}    \\
    \vdots&\vdots&\ddots&\vdots\\
    \theta_1^{k-2} & \theta_2^{k-2} & \cdots & \theta_{2^m-1}^{k-2}  \\
    \theta_1^{k} & \theta_2^{k} & \cdots & \theta_{2^m-1}^{k}  \\
  \end{pmatrix},
\end{aligned}
\end{equation}
where $3\leq k\leq 2^m-4$. In the following, we prove that ${\C_3}$ is an NMDS code with parameters $[2^m-1,k,2^m-1-k]_{2^m}$ supporting $2$-designs and determine the weight distribution of this code. We first consider the parameters of the dual of ${\C_3}$.
\begin{lemma}\label{le12}
Let $\mathcal{C}_3$ be a linear code over ${\bF_{2^m}}$ with generator matrix $G_3$ defined in (\ref{G3}). Then ${\C_3}^\perp$ is a $[2^m-1,2^m-k-1,k]_{2^m}$ AMDS code and the number of its minimum weight codewords is
\begin{equation*}
\begin{aligned}
A_k^\perp=\frac{(2^m-1)^2}{2^m}\left[\frac{1}{k}\binom{2^m-2}{k-1}+(-1)^{k+\lfloor\frac{k}{2}\rfloor}\binom{2^{m-1}-1}{\lfloor\frac{k}{2}\rfloor}\right].
\end{aligned}
\end{equation*}
\end{lemma}
\begin{proof}
We first prove that $\dim({\C_3})=k$. Let $\mathbf{g}_i$ be the $i$-th row of $G_3$, where $1\leq i\leq k$. Assume that there exist $a_1,a_2,\ldots,a_k\in\mathbb{F}_{2^m}$ that are not all $0$ such that $\sum_{i=1}^ka_i\mathbf{g}_i=\mathbf{0}$. It implies that
$f(\theta_\ell)=a_1+a_2\theta_\ell+\cdots+a_{k-1}\theta_\ell^{k-2}+a_k\theta_\ell^k=0$ for $\ell=1,2,\ldots,2^m-1$.
It is easy to find that the polynomial $f(x)=a_1+a_2x+\cdots+a_{k-1}x^{k-2}+a_kx^k$ has at most $k$ zeros in ${\bF_{2^m}^*}$, which contradicts the fact that $\theta_1,\theta_2,\ldots,\theta_{2^m-1}$ are all elements of ${\bF_{2^m}^*}$. Then $\mathbf{d}_1,\mathbf{d}_2,\ldots,\mathbf{d}_k$ are ${\bF_{2^m}}$-linearly independent. Hence, $\dim({\C_3})=k$.
\par
We then prove that ${\C_3}^\perp$ has parameters $[2^m-1,2^m-k-1,k]_{2^m}$. Clearly, $\dim({\C}_3^\perp)=2^m-1-\dim({\C_3})=2^m-k-1$. Below, we need to prove that $d({\C_3}^\perp)=k$.  Following arguments similar to those in Lemma \ref{le5}, we can show that any $k-1$ columns in $G_3$ are ${\bF_{2^m}}$-linearly independent, which implies $d(\mathcal{C}_3^\perp) \geq k$. Now we consider a $k$ by $k$ submatrix $G_{3,1}$ formed by any $k$ columns of $G_3$ as follows:
\begin{equation}\label{G31}
\begin{aligned}
G_{3,1}=
\begin{pmatrix}
    1 & 1 & \cdots &  1  \\
    \theta_{i_1} & \theta_{i_2} & \cdots & \theta_{i_{k}}  \\
    \vdots&\vdots&\ddots&\vdots\\
    \theta_{i_1}^{k-2} & \theta_{i_2}^{k-2} & \cdots & \theta_{i_{k}}^{k-2}  \\
    \theta_{i_1}^{k} & \theta_{i_2}^{k} & \cdots & \theta_{i_{k}}^{k}  \\
   \end{pmatrix}.
\end{aligned}
\end{equation}
By Lemma \ref{le2.6}, we have
$$\det(G_{3,1}) = \left( \prod_{1 \leq a < b \leq k} (\theta_{i_b} - \theta_{i_a}) \right) \cdot \sum_{j=1}^k \theta_{i_j}.$$
By Lemma \ref{le2.5}, the number of distinct $k$-subsets $\{\theta_{i_1},\theta_{i_2},\ldots,\theta_{i_k}\}\subseteq\mathbb{F}_q^*$ such that $\det(G_{3,1})=0$ is equal to $N(k,0,{\bF_{2^m}^*})$, where $3\leq k\leq 2^m-4$ satisfying $N(k,0,{\bF_{2^m}^*})\neq0$. This means that there exist $k$ columns of $G_3$ are ${\bF_q}$-linearly dependent. By following the proof technique of Lemma \ref{le6}, we conclude that $\mathcal{C}_3$ is a $[2^m-1,2^m-k-1,k]_{2^m}$ AMDS code and the number of its minimum weight codewords is
\begin{equation*}
A_k^\perp=(q-1)N(k,0,{\bF_{2^m}^*})
=\frac{(2^m-1)^2}{2^m}\left[\frac{1}{k}\binom{2^m-2}{k-1}+(-1)^{k+\lfloor\frac{k}{2}\rfloor}\binom{2^{m-1}-1}{\lfloor\frac{k}{2}\rfloor}\right],
\end{equation*}
where $3\leq k\leq 2^m-4$.
\end{proof}
Next, we prove that the minimum weight codewords of ${\C_3}^\perp$ support a simple $2$-design.
\begin{lemma}\label{le13}
Let $\mathcal{C}_3$ be a linear code over ${\bF_{2^m}}$ with generator matrix $G_3$ defined in (\ref{G3}). Then the minimum weight codewords of ${\C_3}^\perp$ support a $2$-$(2^m-1,k,\lambda_1^c)$ simple design, where
$$\lambda_1^c=\frac{1}{2^m}\left[\binom{2^m-3}{k-2}+(-1)^{k+\lfloor\frac{k}{2}\rfloor}
\frac{k(k-1)}{2\lfloor\frac{k}{2}\rfloor}\binom{2^{m-1}-2}{\lfloor\frac{k}{2}\rfloor-1}\right].$$
\end{lemma}
\begin{proof}
 Let $\mathbf{c}^\perp=(c_1,c_2,\ldots,c_{q-1})$ be the codeword of ${\C_3}^\perp$ with minimum weight $k$. Then we have $G_3\cdot(c_1,c_2,\ldots,c_{q-1})^\top=\mathbf{0}_{k}$. Precisely, we have the following homogeneous system of equations holds:
$$
G_{3,1}\cdot (c_{i_1},c_{i_2},\ldots,c_{i_k})^\top=\mathbf{0}_k,
$$
where $G_{3,1}$ is given by Eq.(\ref{G31}). Therefore, every minimum weight codewords in ${\C_3}^\perp$ with nonzero coordinates in $\{i_1,i_2,\ldots,i_k\}$ must correspond to the $k$-subset $\{\theta_{i_1},\theta_{i_2},\ldots,\theta_{i_k}\}$. By the coding-theoretic construction of $t$-designs, we define $\mathcal{P}({\C_3}^\perp):=\{1,2,\ldots,q-1\}$ and
\begin{equation*}
\begin{aligned}
\mathcal{B}_k({\C_3}^\perp)&:=\{\!\!\{\operatorname{supp}(\mathbf{c}^\perp) : \mathbf{c}^\bot \in{\C_3}^\perp,\ \operatorname{wt}(\mathbf{c}^\perp) = k\}\!\!\}\\
&=\{\!\!\{i_1^{(1)},i_2^{(1)},\ldots,i_k^{(1)}\},\{i_1^{(2)},i_2^{(2)},\ldots,i_k^{(2)}\},\ldots,\{i_1^{(N)},i_2^{(N)},\ldots,i_k^{(N)}\}\!\!\},
\end{aligned}
\end{equation*}
where $N:=N(k,0,{\bF_{2^m}^*})$ and $3\leq k\leq 2^m-4$. For any two different fixed elements $\theta_{i_1},\theta_{i_2}$, by Corollary \ref{coro5.5} given in the Appendix, the total number of different choices of $\theta_{i_3},\theta_{i_4},\ldots,\theta_{i_k}\in{\bF_{2^m}^*}$ such that $\det(G_{3,1})=0$ is equal to $\frac{k(k-1)N(k,0,{\bF_{2^m}^*})}{(2^m-1)(2^m-2)}$. It means that the number of choices of $\theta_{i_3},\theta_{i_4},\ldots,\theta_{i_k}$ is independent of $\theta_{i_1},\theta_{i_2}$. Then the pair $(\mathcal{P}({\C_3}^\perp),\mathcal{B}_k({\C_3}^\perp))$ is a $2$-$(2^m-1,k,\lambda_1^c)$ simple design with $\#\mathcal{B}_k(\mathcal{C}_3^\perp)=N(k,0,{\bF_{2^m}^*})$  distinct blocks, where $3\leq k\leq 2^m-4$ and
\begin{equation*}
\begin{aligned}
\lambda_1^c&=\frac{k(k-1)N(k,0,{\bF_{2^m}^*})}{(2^m-1)(2^m-2)}=
\frac{1}{2^m}\left[\binom{2^m-3}{k-2}+(-1)^{k+\lfloor\frac{k}{2}\rfloor}
\frac{k(k-1)}{2\lfloor\frac{k}{2}\rfloor}\binom{2^{m-1}-2}{\lfloor\frac{k}{2}\rfloor-1}\right].
\end{aligned}
\end{equation*}
Therefore, we can deduce that the codeword of weight $k$ in ${\C_3}^\perp$ support a $2$-$\left(2^m-1,k,\lambda_1^c\right)$ simple design.
\end{proof}

Finally, we establish that ${\C_3}$ is a $[2^m-1,k,2^m-1-k]_{2^m}$ NMDS code supporting a simple $2$-design.
\begin{theorem}\label{th4.1}
Let $\mathcal{C}_3$ be a linear code over ${\bF_{2^m}}$ with generator matrix $G_3$ defined in (\ref{G3}). Then ${\C_3}$ is an NMDS code with parameters $[2^m-1,k,2^m-1-k]_{2^m}$, has exactly
$$
A_{2^m+1-k} = \frac{(2^m-1)^2}{2^m}\left[\frac{1}{k}\binom{2^m-2}{k-1}+(-1)^{k+\lfloor\frac{k}{2}\rfloor}\binom{2^{m-1}-1}{\lfloor\frac{k}{2}\rfloor}\right]
$$
minimum weight codewords, and its weight distribution is determined by Lemma \ref{le2.1}.
Moreover, the minimum weight codewords of ${\C_3}$ support a $2$-$(2^m-1,2^m-1-k,\lambda_1)$ simple design, where
$$\lambda_1=\frac{(2^m-1-k)(2^m-2-k)}{2^mk(k-1)}\left[\binom{2^m-3}{k-2}+(-1)^{k+\lfloor\frac{k}{2}\rfloor}
\frac{k(k-1)}{2\lfloor\frac{k}{2}\rfloor}\binom{2^{m-1}-2}{\lfloor\frac{k}{2}\rfloor-1}\right].$$
\end{theorem}
\begin{proof}
The generator matrix of ${\C_3}$ given by Eq.(\ref{G3}) produces codewords through linear combinations of its rows, yielding the code structure
$${\C_3}=\left\{\mathbf{c}=(a_1+a_2x+\cdots+a_{k-1}x^{k-2}+a_kx^k)_{x\in\mathbb{F}_{2^m}^*}:a_1,a_2,\ldots,a_k\in{\bF_{2^m}}\right\}.$$
The polynomial $f(x)=a_1+a_2x+\cdots+a_{k-1}x^{k-2}+a_kx^k$ possesses at most $k$ roots in ${\bF_{2^m}^*}$, which implies that $d({\C_3})\geq 2^m-1-k$. By the Singleton bound, $d({\C_3})\leq 2^m-k$. If $d({\C_3})=2^m-k$, then ${\C_3}$ is MDS, so $\mathcal{C}_3^\perp$ must also be MDS. However, Lemma \ref{le12} already shows that $\mathcal{C}_3^\perp$ is AMDS with $d(\mathcal{C}_3^\perp)=k$, contradicting the MDS property.
Consequently, $d({\C_3})=2^m-k-1$, confirming ${\C_3}$ as an NMDS code with parameters $[2^m-1,k,2^m-k-1]_{2^m}$ for $3\leq k\leq 2^m-4$. Lemma \ref{le2.2} establishes the equality $A_{2^m-k-1}=A_k^\bot$, while Lemma \ref{le2.1} directly provides the weight distribution of ${\C_3}$. Through Lemmas \ref{le2.2}, \ref{le13} and Eq.(\ref{eq(1)}), the minimum weight codewords of ${\C_3}$ support a $2$-$(2^m-1,2^m-1-k,\lambda_1)$ simple design with
\begin{equation*}
\begin{aligned}
\lambda_1&=\frac{\binom{2^m-1-k}{2}}{(2^m-1)\binom{2^m-1}{2}}\cdot A_{2^m-1-k}\\
&=\frac{(2^m-1-k)(2^m-2-k)}{2^m(k-1)k}\left[\binom{2^m-3}{k-2}+(-1)^{k+\lfloor\frac{k}{2}\rfloor}
\frac{k(k-1)}{2\lfloor\frac{k}{2}\rfloor}\binom{2^{m-1}-2}{\lfloor\frac{k}{2}\rfloor-1}\right].
\end{aligned}
\end{equation*}
Moreover, by Lemma \ref{le13}, $\lambda_1^c=\binom{2^m-3}{2^m-1-k}/\binom{2^m-3}{2^m-3-k}\cdot\lambda_1$. By Eq.(\ref{eq(1.1)}), the pair $(\mathcal{P}({\C_3}^\perp), \mathcal{B}_k({\C_3}^\perp))$ forms a complementary design to $(\mathcal{P}({\C_3}), \mathcal{B}_{2^m-1-k}({\C_3}))$.
\end{proof}

Below, we present an example that support the results given in Lemma \ref{le13} and Theorem \ref{th4.1}, which have been verified using Magma programs.
\begin{example}
Let $m=3,k=4$ and $\omega$ be a primitive element of ${\bF_{2^3}}$. Then
$$
G_3=
\left(
  \begin{array}{ccccccc}
    1 & 1 & 1 & 1 & 1 & 1 & 1 \\
    1 & \omega & \omega^2 & \omega^3 & \omega^4 & \omega^5 & \omega^6 \\
    1 & \omega^2 & \omega^4 & \omega^6 & \omega & \omega^3 & \omega^5 \\
    1 & \omega^4 & \omega & \omega^5 & \omega^2 & \omega^6 & \omega^3 \\
  \end{array}
\right).
$$
By Magma, ${\C_3}$ generated by $G_3$ is a $[7,4,3]_{2^3}$ NMDS code with weight enumerator $A(x)=1+49x^3+49x^4+882x^5+1740x^6+1645x^7.$ Moreover, $\mathcal{P}({\C_3}):=\{1,2,3,4,5,6,7\}$ and
$
\mathcal{B}_3({\C_3}):=\{\!\!\{\operatorname{supp}(\mathbf{c}):\mathbf{c}\in {\C_3},\ \operatorname{wt}(\mathbf{c})=3\}\!\!\}
=\{\!\!\{3,4,6\},\{1,2,4\},\{2,3,5\},\{1,5,6\},\{4,5,7\},\{2,6,7\},\{1,3,\\7\}\!\!\}.
$
By the definition of a $t$-designs, every pair of elements in $\mathcal{P}({\C_3})$ is contained in exactly $\lambda_1 = 1$ block of $\mathcal{B}_3({\C_3})$. Consequently, the pair $(\mathcal{P}({\C_3}), \mathcal{B}_3({\C_3}))$ supports a $2\text{-}(7,3,1)$ simple design. Analogously, $\mathcal{P}(\mathcal{C}_3^\perp):=\{1,2,3,4,5,6,7\}$ and
$
\mathcal{B}_4(\mathcal{C}_3^\perp):=\{\!\!\{\operatorname{supp}(\mathbf{c}^\perp) : \mathbf{c}^\perp \in\mathcal{C}_3^\perp,\ \operatorname{wt}(\mathbf{c}^\perp) = 4\}\!\!\}
=\{\!\!\{1, 2, 3, 6\},\{2, 3, 4, 7\},\{1, 3, 4, 5\},\{3, 5, 6, 7\},\{ 1, 2, 5, 7\},\{2, 4, 5, 6\},\{1, 4,6, 7\}\!\!\}.
$
By the definition of a $t$-design, every pair of elements in $\mathcal{P}({\C_3}^\perp)$ is contained in exactly $\lambda_1^c = 2$ blocks of $\mathcal{B}_4({\C_3}^\perp)$. Consequently, the pair $(\mathcal{P}({\C_3}^\perp), \mathcal{B}_4({\C_3}^\perp)$ supports a $2\text{-}(7,4,2)$ simple design.
\end{example}
\subsection{New NMDS codes with parameters $[2^m,k,2^m-k]_{2^m}$ supporting $3$-designs}
Let ${\C_4}$ denote the fourth family of linear codes over ${\bF_{2^m}}$ with generator matrix
\begin{equation}\label{G4}
\begin{aligned}
G_4=
\begin{pmatrix}
    1 &          \cdots & 1               &1\\
    \theta_1      & \cdots & \theta_{2^m-1}&0    \\
    \vdots&\ddots&\vdots&\vdots\\
    \theta_1^{k-2}  & \cdots & \theta_{2^m-1}^{k-2} &0 \\
    \theta_1^{k}  & \cdots & \theta_{2^m-1}^{k}&0  \\
\end{pmatrix},
\end{aligned}
\end{equation}
where $3\leq k\leq 2^m-3$. By employing proof methods similar to those used in Lemma \ref{le12}, Lemma \ref{le13}, and Theorem \ref{th4.1}, we can derive the fourth infinite family of NMDS codes with parameters $[2^m,k,2^m-k]_q$ supporting a simple $3$-design.
\begin{theorem}\label{th4.4}
Let $\mathcal{C}_4$ be a linear code over ${\bF_{2^m}}$ with generator matrix $G_4$ defined in (\ref{G4}). Then ${\C_4}$ is an NMDS code with parameters $[2^m,k,2^m-k]_{2^m}$, has exactly
$$A_{2^m-k}=A_k^\bot=\left\{
                       \begin{array}{ll}
                         \frac{2^m-1}{2^m}\binom{2^m}{k}, & \hbox{if $k$ is odd;} \\
                         \frac{(2^m-1)^2}{2^m}\left[\frac{2^m}{k(k-1)}\binom{2^m-2}{k-2}+(-1)^{\frac{3k}{2}}\binom{2^{m-1}}{\frac{k}{2}}\right], & \hbox{if $k$ is even}
                       \end{array}
                     \right.
$$
minimum weight codewords, and its weight enumerator is determined by Lemma \ref{le2.1}.
Moreover, if $4\leq k\leq 2^m-4$ is even, then the minimum weight codewords of ${\C_4}$ support a $3$-$(2^m,2^m-k,\lambda_2)$ simple design and the minimum weight codewords of ${\C_4}^\perp$ support a $3$-$(2^m,k,\lambda_2^c=\frac{\lambda_2k(k-1)(k-2)}{(2^m-k)(2^m-k-1)(2^m-k-2)})$ simple design, where
$$\lambda_2=\frac{(2^m-k)(2^m-k-1)(2^m-k-2)}{2^mk(k-1)(k-2)}\left[\binom{2^m-3}{k-3}+(-1)^{\frac{k}{2}}(k-1)\binom{2^{m-1}-2}{\frac{k}{2}-2}\right].$$
\end{theorem}

Below, we present an example that support the results given in Theorem \ref{th4.4}, which have been verified using Magma programs.
\begin{example}
Let $m=3,k=4$ and $\omega$ be a primitive element of ${\bF_{2^3}}$. Then
$$
G_4=
\left(
  \begin{array}{cccccccc}
    1 & 1 & 1 & 1 & 1 & 1 & 1&1 \\
    1 & \omega & \omega^2 & \omega^3 & \omega^4 & \omega^5 & \omega^6 &0\\
    1 & \omega^2 & \omega^4 & \omega^6 & \omega & \omega^3 & \omega^5 &0\\
    1 & \omega^4 & \omega & \omega^5 & \omega^2 & \omega^6 & \omega^3 &0\\
  \end{array}
\right).
$$
By Magma, ${\C_4}$ generated by $G_4$ is a $[8,4,4]_{2^3}$ NMDS code with weight enumerator $A(x)=1+98x^4+1176x^6+1344x^7+1477x^8.$ Moreover, $\mathcal{P}({\C_4}):=\{1,2,3,4,5,6,7,8\}$ and
$
\mathcal{B}_4({\C_4}):=\{\!\!\{\operatorname{supp}(\mathbf{c}) : \mathbf{c} \in {\C_4},\ \operatorname{wt}(\mathbf{c}) = 4\}\!\!\}
=\{\!\!\{1,3,7,8\},\{2,4,5,6\},\{1,2,3,6\},\{3,5,6,7\},\{1,2,4,8\},\{1,\\4,6,7\},\{4,5,7,8\},\{2,3,4,7\},\{1,5,6,8\},\{1,3,4,5\},
\{1,2,5,7\},\{3,4,6,8\},\{2,6,7,8\},\{2,3,5,\\8\}\!\!\}.
$
By the definition of $t$-designs, every $3$-subset in $\mathcal{P}({\C_4})$ is contained in exactly $\lambda = 1$ block of $\mathcal{B}_4(\mathcal{C}_4)$. Consequently, the pair $(\mathcal{P}({\C_4}), \mathcal{B}_4({\C_4}))$ supports a $3\text{-}(8,4,1)$ simple design. Moreover, the pair $(\mathcal{P}({\C_4}^\perp), \mathcal{B}_4({\C_4}^\perp))$ supports the same $3\text{-}(8,4,1)$ simple design.
\end{example}

\begin{remark}
We compare the NMDS codes constructed in Theorems \ref{th4.1} and \ref{th4.4} with the known ones in \cite{Heng2023,Heng20231,Xu2022} as follows:
\begin{itemize}
  \item Setting $k=3,5$ and $6$ in Theorems \ref{th4.1} and \ref{th4.4}, we can directly obtain the corresponding results of Theorems 18, 27, 35 and 38 of \cite{Heng2023}, respectively.
  \item Setting $k=4$ in Theorem \ref{th4.4}, we can directly obtain the corresponding results for $h=1$ studied in \cite[Theorem 7]{Xu2022} and $h=2$ studied in
  \cite[Theorem 15]{Heng20231}.
  \item Combining the results of Lemma \ref{le13} with those in Theorem \ref{th4.1}, we provide an affirmative answer to the conjecture of Heng and Wang (See \cite[Conjecture36]{Heng2023}).
\end{itemize}

To our knowledge, these codes constructed in Theorem \ref{th4.1} (resp. Theorem \ref{th4.4}) are the first infinite family of NMDS codes with arbitrary dimensions supporting $2$-designs (resp. $3$-designs).
\end{remark}

\section{Summary and concluding remarks}\label{sec.5}

This paper presented two new infinite families of NMDS codes over $\mathbb{F}_q$ of lengths $q+1$ and $q+2$ with arbitrary dimensions for any  prime power $q$. These codes generalized the related constructions of NMDS codes given in \cite{DingY2024}. Compared with the known ones constructed in \cite{Ding2020,Heng2022,Li2023,Wang2021}, they had different generator matrices and weight distributions, and more flexible parameters. Moreover, we established two infinite families of NMDS codes over $\mathbb{F}_{2^m}$ with arbitrary dimensions supporting $2$-designs and $3$-designs. To our knowledge, these constitute the first infinite families of NMDS codes with arbitrary dimensions admitting $t$-designs for $t\geq2$. Notably, our results not only generalize fixed-dimension constructions in \cite{Heng2023,Heng20231,Xu2022}, but also provide an affirmative resolution to the Heng-Wang conjecture \cite{Heng2023}.

Building on the work of Tang and Ding \cite{Tang2021}, who constructed the first infinite family of NMDS codes with length $q+1$ and dimension 6 holding $4$-designs over $\mathbb{F}_{2^m}$, we propose the following open problem: \emph{Does there exist an infinite family of NMDS codes with arbitrary dimensions supporting $4$-designs?}

\newpage

\section*{Appendix}\label{Appendix}

The central challenge in constructing NMDS codes from specialized generator matrices with arbitrary dimensions that support $t$-designs lies in determining their weight distributions through the computation of
$$N(k,0,D) = \#\left\{\{x_1,\ldots,x_k\} \subseteq D : \sum_{i=1}^k x_i = 0\right\},$$
where $ D = \mathbb{F}_q $ or $\mathbb{F}_q^*$. Recent studies \cite{DingY2024,Heng2023,Li2023,Xu2023} explicitly evaluated $N(k,0,D)$ for $3\leq k\leq 6$. Notably, this problem is equivalent to the classical subset sum problem, which was systematically resolved for a general $k$ by Li and Wan~\cite{Li2008}, with further refinements provided by Pavone \cite{Pavone2021}. During the initial preparation of this work, we independently sought to generalize the reulst of $N(k,0,D)$ for $3\leq k\leq 6$ in \cite{DingY2024,Heng2023,Li2023,Xu2023}  to a general case, prior to recognizing the broader solutions in \cite{Li2008,Pavone2021}.

Our contribution includes an alternative algebraic method for computing $N(k,0,D)$ in the binary case $(D= \mathbb{F}_{2^m} \,\,{\rm or}\,\, \mathbb{F}_{2^m}^* )$, yielding formulas equivalent to those in Theorem 1.2 of \cite{Li2008}. Beyond reproducing known results, we establish novel recurrence relations for $N(k,0,\mathbb{F}_{2^m})$  and  $N(k,0,\mathbb{F}_{2^m}^*)$, as well as a combinatorial identity that may facilitate future analyses.

First, we establish two new recurrence relations for $N(k,0,\mathbb{F}_{2^m}^*)$.

\begin{lemma}\label{th5.1}
Let $\{x_1,\ldots,x_k\}\subseteq\mathbb{F}_{2^m}^*$ be a $k$-subset, where $ 3 \leq k < 2^m- 1 $.
Then the following results hold.

(i) For odd $ k = 2\ell + 1 $ ($ \ell \geq 1 $),
\begin{equation}\label{eq(6)}
N(2\ell+1,0,{\mathbb{F}_{2^m}^*}) + \frac{(2^m-2\ell)}{2\ell} N(2\ell-1,0,{\mathbb{F}_{2^m}^*}) = \frac{1}{2^m} \binom{2^m}{2\ell+1}.
\end{equation}

(ii) For even $ k = 2\ell $ ($ \ell \geq 2 $),
\begin{equation}\label{eq(7)}
N(2\ell,0,{\mathbb{F}_{2^m}^*}) = \frac{(2^m-2\ell)}{2\ell}N(2\ell-1,0,{\mathbb{F}_{2^m}^*}).
\end{equation}
\end{lemma}

\begin{proof}
Let $N(k,0,{\bF_{2^m}})$ denote the number of distinct $k$-element subsets $ \{ x_1, x_2, \ldots, x_k\}  \subseteq \mathbb{F}_{2^m} $  with $ \sum_{i=1}^k x_i = 0 $. We partition these subsets based on whether they contain $0$:
\begin{equation}\label{eq(7.1)}
\begin{aligned}
N(k,0,{\bF_{2^m}}) = N(k,0,{\bF_{2^m}^*}) + N(k-1,0,{\bF_{2^m}^*}).
\end{aligned}
\end{equation}
Let $ \mathbb{F}_{2^m}^* = \langle \omega \rangle $ with basis $\{1, \omega, \omega^2, \ldots, \omega^{m-1}\}$ over $ \mathbb{F}_2 $. Under this basis, each $ x_i \in \mathbb{F}_{2^m} $ corresponds to a vector $\mathbf{x}_i = (x_{i,0}, x_{i,1}, \ldots, x_{i,m-1})^\top \in \mathbb{F}_2^m $, where $ x_i = \sum_{j=0}^{m-1} x_{i,j} \omega^j $, and a $k$-element subset $\mathcal{X}=\{ x_1, x_2, \ldots, x_k\} \subseteq \mathbb{F}_{2^m}$ corresponds to an $m\times k $ binary matrix as follows:
$$
M(\mathcal{X})
= \begin{pmatrix}
	x_{1,0} & x_{2,0} & \cdots & x_{k,0} \\
	x_{1,1} & x_{2,1} & \cdots & x_{k,1} \\
	\vdots & \vdots & \ddots & \vdots \\
	x_{1,m-1} & x_{2,m-1} & \cdots & x_{k,m-1}
\end{pmatrix}
=\left(
\begin{array}{c}
	\mathbf{r}_0 \\
	\mathbf{r}_1 \\
	\vdots \\
	\mathbf{r}_{m-1} \\
\end{array}
\right)
$$
We define a row parity of $M(\mathcal{X})$ as follows:
\begin{equation*}
\begin{aligned}
 \delta(\mathcal{X})&=(\delta(\mathbf{r}_0),\delta(\mathbf{r}_1),\ldots,\delta(\mathbf{r}_{m-1}))\in \mathbb{F}_2^m,
\end{aligned}
\end{equation*}
where $\delta(\mathbf{r}_j) = \sum_{i=1}^{k} x_{i,j} \,\, {\rm mod }\,\, 2, \,\, 0\leq j\leq m-1.$
The row parity define equivalent relation on the set of matrices $M(\mathcal{X})$ for all $k$-element subsets. Two matices are \emph{equivalent} if they have identical row parities. This partitions  all $m\times k $ binary matrices into $2^m$ equivalence classes since
$\mathbb{F}_2^m$ has $2^m$ distinct vectors.

For odd $k$, we show that all classes have equal size. Let $A(\delta(\mathcal{X}))$ denote the class of mathices with row parity $\delta(\mathcal{X})$ for a $k$-elements subset. Let $\bar{x}_{ij} = x_{ij}+ 1 \,\, {\rm mod} \,\, 2$, and $\bar{\mathbf{r}_j}=(\bar{x}_{1,j}, \bar{x}_{2,j},\ldots, \bar{x}_{k,j})$ denote the component-wise complement of the $j$-th row $\mathbf{r}_j$ of $M(\mathcal{X})$ for $0 \leq j \leq m-1$. Since $k$ is odd, we have
\begin{equation}\label{eq:parity-flip6}
\delta(\bar{\mathbf{r}_j}) = \delta(\mathbf{r}_j) + 1 \quad {\rm mod} \,\, 2.
\end{equation}
Consider two distinct equivalence classes $A(\delta(\mathcal{X}_1))$ and $A(\delta(\mathcal{X}_2))$, where $\delta(\mathcal{X}_1)\neq \delta(\mathcal{X}_2) \in \mathbb{F}_2^m$. Let $J = \{j \mid \delta(\mathbf{r}_j^{(1)}) \neq \delta(\mathbf{r}_j^{(2)})\}$ index the differing parity coordinates, where $\delta(\mathbf{r}_j^{(t)})$ denotes the $j$-th coordinates of $\delta(\mathcal{X}_t)$ with $t=1,2$. For any $M(\mathcal{X}_1) \in A(\delta(\mathcal{X}_1))$, apply the component-wise complement operation to rows $j \in J$ of $M(\mathcal{X}_1)$. By the parity-flip property \eqref{eq:parity-flip6}, this maps $M(\mathcal{X}_1)$ to $M(\mathcal{X}_2) \in A(\delta(\mathcal{X}_2))$. Since the operation is invertible (applying it twice recovers $M(\mathcal{X}_1)$), it induces a bijection between $A(\delta(\mathcal{X}_1))$ and $A(\delta(\mathcal{X}_2))$. Thus, all equivalence classes have equal size:
\[
\#A(\delta(\mathcal{X})) = \frac{1}{2^m} \binom{2^m}{k}, \quad \forall \delta(\mathcal{X}) \in \mathbb{F}_2^m.
\]
Specially, for $\delta(\mathcal{X}) = \mathbf{0}$,
$
N(k,0,\mathbb{F}_{2^m}) = \#A(\mathbf{0}) = \frac{1}{2^m} \binom{2^m}{k}.
$
By the recurrence relation (\ref{eq(7.1)}), for odd $k = 2\ell + 1$ with $\ell \geq 1$,
\begin{equation}\label{eq(8)}
N(2\ell+1,0,\mathbb{F}_{2^m}^*) + N(2\ell,0,\mathbb{F}_{2^m}^*) = \frac{1}{2^m} \binom{2^m}{2\ell+1}.
\end{equation}

For even \( k = 2\ell \) (\(\ell \geq 2\)), consider subsets \( \mathcal{X} \subseteq \mathbb{F}_{2^m} \) with \( \sum_{i=1}^{2\ell} x_i = 0 \). Let \( N'(2\ell,0,{\mathbb{F}_{2^m}}) \) denote the number of \(2\ell\)-element subsets containing \(0\), and \( N''(2\ell,0,{\mathbb{F}_{2^m}}) \) the number of subsets with all elements nonzero. Their enumeration proceeds as follows:
\begin{itemize}
  \item ~If \( 0 \in \mathcal{X} \), w.l.o.g. assume \( x_{2\ell} = 0 \). Then \( \sum_{i=1}^{2\ell-1} x_i = 0 \), where all \( x_i \in \mathbb{F}_{2^m}^* \). Each configuration is counted \( 2\ell \) times (choosing which element is zero), giving
\[
N'(2\ell,0,\mathbb{F}_{2^m}) = \frac{N(2\ell-1,0,\mathbb{F}_{2^m}^*)}{2\ell}.
\]
  \item  ~If \( x_i \neq 0 \), w.l.o.g. let \( x_{2\ell} \neq 0 \). Then \( \sum_{i=1}^{2\ell-1} x_i = x_{2\ell} \). Normalizing by \( x_{2\ell}^{-1} \), set \( y_i = x_i/x_{2\ell} \) (\( y_i \neq 1 \)) to get \( \sum_{i=1}^{2\ell-1} y_i = 1 \). Substituting \( z_i = y_i + 1 \) yields
\[
\sum_{i=1}^{2\ell-1} z_i = 0, \quad z_i \in \mathbb{F}_{2^m}^*,
\]
with solution count \( N(2\ell-1,0,\mathbb{F}_{2^m}^*) \). Scaling by \( 2^m - 1 \) (choices for \( x_{2\ell} \)) and dividing by \( 2\ell \) gives
\[
N''(2\ell,0,\mathbb{F}_{2^m}) = \frac{(2^m - 1)N(2\ell-1,0,\mathbb{F}_{2^m}^*)}{2\ell}.
\]
\end{itemize}
Combining both cases:
\[
N(2\ell,0,\mathbb{F}_{2^m}^*) + N(2\ell-1,0,\mathbb{F}_{2^m}^*) = \frac{2^m N(2\ell-1,0,\mathbb{F}_{2^m}^*)}{2\ell}.
\]
Substituting into \eqref{eq(8)} yields the recurrence:
\[
N(2\ell+1,0,\mathbb{F}_{2^m}^*) + \frac{2^m - 2\ell}{2\ell}N(2\ell-1,0,\mathbb{F}_{2^m}^*) = \frac{1}{2^m}\binom{2^m}{2\ell+1}.
\]

\end{proof}

In the following, we obtain explicit formulas for $N(k,0,{\bF_{2^m}^*})$ and $N(k,0,{\bF_{2^m}})$ by using the recurrence relations given in Lemma \ref{th5.1}.
\begin{theorem}\label{th5.2}
Follow the notation and conditions of Lemma \ref{th5.1}. For $\ell\in\mathbb{Z}^+$, the following results hold.
\par
${\rm (i)}$~For odd $ k = 2\ell + 1 $ ($ \ell \geq 1 $),
\begin{equation}\label{eq(12)}
\begin{aligned}
N(2\ell+1,0,{\bF_{2^m}^*}) = \frac{1}{{2^m}} \sum_{t=0}^{\ell-1} (-1)^t \left( \prod_{j=0}^{t-1} \frac{{2^m} - 2(\ell - j)}{2(\ell - j)} \right) \binom{{2^m}}{2\ell+1 - 2t}.
\end{aligned}
\end{equation}
\par
${\rm (ii)}$~For even $ k = 2\ell $ ($ \ell \geq 2 $),
\begin{equation}\label{eq(13)}
\begin{aligned}
N(2\ell,0,{\bF_{2^m}^*}) = \frac{2^m-2\ell}{2^{m+1}\ell} \sum_{t=0}^{\ell-2} (-1)^t \left( \prod_{j=0}^{t-1} \frac{2^m - 2(\ell - j-1)}{2(\ell - j-1)} \right) \binom{2^m}{2\ell-1 - 2t}.
\end{aligned}
\end{equation}
\end{theorem}
\begin{proof}
By the recurrence relation (\ref{eq(6)}), it is easy to see that
\begin{equation}\label{eq(14)}
\begin{aligned}
\left\{
  \begin{array}{ll}
    N(2\ell+1,0,{\bF_{2^m}^*})=\frac{1}{2^m}\binom{2^m}{2\ell+1}-\frac{2^m-2\ell}{2\ell}\cdot N(2\ell-1,0,{\bF_{2^m}^*}), & \hbox{} \\
    N(2\ell-1,0,{\bF_{2^m}^*})=\frac{1}{2^m}\binom{2^m}{2\ell-1}-\frac{2^m-2\ell+2}{2\ell-2}\cdot N(2\ell-3,0,{\bF_{2^m}^*}), & \hbox{} \\
    %N(2\ell-3,0,{\bF_{2^m}^*})=\frac{1}{2^m}\binom{2^m}{2\ell-3}-\frac{2^m-2\ell+4}{2\ell-4}\cdot %N(2\ell-5,0,{\bF_{2^m}^*}), & \hbox{} \\
    ~~~~~~~~~~~~~~~~~~\vdots & \hbox{} \\
   % N(7,0,{\bF_{2^m}^*})=\frac{1}{2^m}\binom{2^m}{7}-\frac{2^m-6}{6}\cdot N(5,0,{\bF_{2^m}^*}), & \hbox{} \\
    N(5,0,{\bF_{2^m}^*})=\frac{1}{2^m}\binom{2^m}{5}-\frac{2^m-4}{4}\cdot N(3,0,{\bF_{2^m}^*}), & \hbox{} \\
    N(3,0,{\bF_{2^m}^*})=\frac{1}{2^m}\binom{2^m}{3}. & \hbox{} \\
  \end{array}
\right.
\end{aligned}
\end{equation}
We next use mathematical induction to prove that the result (i) holds. For \( \ell = 1 \), we have
\[
N(3,0,{\bF_{2^m}^*}) = \frac{1}{2^m}\binom{2^m}{3}.
\]
Let the empty product (when \( t = 0 \)) be interpreted as 1. Then the proposed formula for \( \ell = 1 \) gives
\[
N(3,0,{\bF_{2^m}^*}) = \frac{1}{2^m} \sum_{t=0}^{0} (-1)^0\cdot1\cdot \binom{2^m}{3 - 0} = \frac{1}{2^m}\binom{2^m}{3}.
\]
Thus, the case with $\ell=1$ holds.

Assume the formula holds for \( \ell = s \), i.e.,
\[
N(2s+1,0,{\bF_{2^m}^*}) = \frac{1}{2^m} \sum_{t=0}^{s-1} (-1)^t \left( \prod_{j=0}^{t-1} \frac{2^m - 2(s - j)}{2(s - j)} \right) \binom{2^m}{2s+1 - 2t}.
\]
Using the recurrence relation for \( N(2(s+1)+1,0,{\bF_{2^m}^*}) = N(2s+3,0,{\bF_{2^m}^*}) \), then
\[
N(2s+3,0,{\bF_{2^m}^*}) = \frac{1}{2^m}\binom{2^m}{2s+3} - \frac{2^m - 2(s+1)}{2(s+1)}N(2s+1,0,{\bF_{2^m}^*}).
\]
Substitute the inductive hypothesis into the equation, then
\begin{equation*}
\begin{aligned}
&N(2s+3,0,{\bF_{2^m}^*})\\
&= \frac{1}{2^m} \left[ \binom{2^m}{2s+3} + \sum_{t=0}^{s-1} (-1)^{t+1} \left( \frac{2^m - 2(s+1)}{2(s+1)} \prod_{j=0}^{t-1} \frac{2^m - 2(s - j)}{2(s - j)} \right) \binom{2^m}{2s+1 - 2t} \right].
\end{aligned}
\end{equation*}
Apply the change of variable \( t' = t + 1 \) to the summation,
\[
\sum_{t'=1}^{s} (-1)^{t'} \left( \prod_{j=0}^{t'-1} \frac{2^m - 2((s+1) - j)}{2((s+1) - j)} \right) \binom{2^m}{2s+3 - 2t'}.
\]
Combine this with the first term (\( t = 0 \)),
\[
N(2s+3,0,{\bF_{2^m}^*}) = \frac{1}{2^m} \sum_{t=0}^{s} (-1)^t \left( \prod_{j=0}^{t-1} \frac{2^m - 2((s+1) - j)}{2((s+1) - j)} \right) \binom{2^m}{2(s+1)+1 - 2t}.
\]
This matches the proposed formula for \( \ell = s + 1 \). Hence, the inductive step holds.

By mathematical induction, the closed-form expression for \( N(2\ell+1,0,{\bF_{2^m}^*}) \) holds for all integers \( \ell \geq 1 \). Moreover, the result (ii) can be directly obtained from the result (i) and recurrence relation (\ref{eq(7)}).
\end{proof}

From the recurrence relations (\ref{eq(7)}), (\ref{eq(7.1)}), (\ref{eq(8)}) and Eq.(\ref{eq(12)}), we can directly obtain explicit formulas for $N(k,0,\mathbb{F}_{2^m})$.

\begin{theorem}\label{th5.3}
Follow the notation and conditions of Lemma \ref{th5.1}. Define
$$N(k,0,{\bF_{2^m}})=\#\left\{\{x_1,x_2,\ldots,x_k\}\subseteq\mathbb{F}_{2^m}:\sum_{i=1}^kx_i=0\right\},$$
where $3\leq k< 2^m-1$. For $\ell\in\mathbb{Z}^+$,
\begin{equation}\label{eq(15)}
\begin{aligned}
N(k,0,{\bF_{2^m}})=
\left\{
  \begin{array}{ll}
    \frac{1}{2^m}\binom{2^m}{2\ell+1}, & \hbox{if $k=2\ell+1,\ell\geq1$;} \\
    \frac{1}{2\ell}\sum_{t=0}^{\ell-2}(-1)^t\left(\prod_{j=0}^{t-1}\frac{2^m-2(\ell-1-j)}{2(\ell-1-j)}\right)\binom{2^m}{2\ell-1-2t}, & \hbox{if $k=2\ell,\ell\geq2$.}
  \end{array}
\right.
\end{aligned}
\end{equation}
\end{theorem}

\begin{remark}
It is easy to check that the value of $N(2\ell+1,0,\mathbb{F}_{2^m}^*)$ and $N(2\ell-1,0,\mathbb{F}_{2^m}^*)$ given in \cite[Theorem 1.2]{Li2008} satisfy the recurrence relations (\ref{eq(6)}) and (\ref{eq(7)}). On the another hand, Theorems \ref{th5.2} and \ref{th5.3} present the explicit formulas for $N(k,0,{\bF_{2^m}^*})$ and $N(k,0,{\bF_{2^m}})$ by using the recurrence relations (\ref{eq(6)}) and (\ref{eq(7)}) given in Lemma \ref{th5.1}. Therefore, we conclude that the results of Theorems \ref{th5.2} and \ref{th5.3} are equivalent to those in \cite[Theorem 1.2]{Li2008} for binary case.

Combining our results given in Theorems \ref{th5.2} and \ref{th5.3}  with those in \cite[Theorem 1.2]{Li2008}, we can directly derive a combinatorial identity as follows:
\begin{equation*}
\begin{aligned}
  \sum_{t=0}^{\ell-1}(-1)^t\left(\prod_{j=0}^{t-1}\frac{2^m-2(\ell-j)}{2(\ell-j)}\right)\binom{2^m}{2\ell+1-2t}=
  \binom{2^m-1}{2\ell+1}+(-1)^{\ell+1}(2^m-1)\binom{2^{m-1}-1}{\ell},
\end{aligned}
\end{equation*}
where $\ell\geq1$ is a positive integer. Furthermore, the identity may be instrumental in addressing combinatorial problems characterized by symmetry, binary structures, or recursive partitioning.
\end{remark}

Combining with the recurrence relation (\ref{eq(7)}) and Eqs.\eqref{eq(12)}--\eqref{eq(15)}, we directly obtain the following results.
\begin{corollary}\label{coro5.5}
Follow the notation and conditions of Lemma \ref{th5.1}. Then the following results hold.
\begin{itemize}
  \item [(i)] For any integer $k$ with $3\leq k\leq 2^m-4$,
  $$\frac{(2^m-1)(2^m-2)}{k(k-1)}|N(k,0,{\bF_{2^m}^*}).$$
  \item [(ii)] For even integer $k$ with $4\leq k\leq 2^m-4$,
  $$\frac{2^m(2^m-1)(2^m-2)}{k(k-1)(k-2)}|N(k,0,{\bF_{2^m}}).$$
\end{itemize}
\end{corollary}
\end{document}